\documentclass[11pt]{article}
\usepackage[letterpaper,margin=1in]{geometry}
\usepackage{amsmath,amssymb,amsthm,mathtools}
\usepackage{enumitem}
\usepackage{microtype}
\usepackage{authblk}
\usepackage[colorlinks=true,linkcolor=blue,citecolor=blue,urlcolor=blue]{hyperref}
\usepackage[capitalise]{cleveref}

\newtheorem{theorem}{Theorem}[section]
\newtheorem{lemma}[theorem]{Lemma}
\newtheorem{proposition}[theorem]{Proposition}
\newtheorem{corollary}[theorem]{Corollary}
\newtheorem{definition}[theorem]{Definition}
\newtheorem{remark}[theorem]{Remark}

\newcommand{\A}{\mathcal A}
\newcommand{\PermO}{\mathcal P}
\newcommand{\Hh}{\mathcal H}
\newcommand{\Ll}{\mathcal L}
\newcommand{\I}{\mathcal I}
\newcommand{\T}{\mathcal T}

\newcommand{\dom}{\operatorname{dom}}
\newcommand{\im}{\operatorname{im}}
\newcommand{\ran}{\operatorname{ran}}
\newcommand{\Span}{\operatorname{span}}
\newcommand{\cpO}{\mathsf{cpO}}
\newcommand{\ccO}{\mathsf{ccO}}
\newcommand{\Comp}{\mathsf C}
\newcommand{\wtV}{\widetilde V}
\newcommand{\Ddown}{B^{\mathrm{dom}}}
\newcommand{\Idown}{B^{\mathrm{im}}}
\newcommand{\ket}[1]{\lvert #1\rangle}
\newcommand{\bra}[1]{\langle #1\rvert}
\newcommand{\braket}[2]{\langle #1\mid #2\rangle}
\newcommand{\proj}[1]{\ket{#1}\!\bra{#1}}
\newcommand{\norm}[1]{\left\lVert #1\right\rVert}
\newcommand{\ip}[2]{\left\langle #1,#2\right\rangle}
\newcommand{\one}{\mathbf 1}
\newcommand{\Ccont}[1]{\mathsf c(#1)}

\title{Compressed Permutation Oracles Revisited}
\author[1]{Joseph Carolan}
\author[2]{Christian Majenz}
\affil[1]{University of Maryland}
\affil[2]{Technical University of Denmark}
\date{}

\begin{document}
\maketitle

\begin{abstract}
    The compressed permutation oracle has been used to analyze the quantum security of a number of cryptographic constructions which resisted prior techniques. However, these analyses were fundamentally limited by the poor soundness of the method: the technique was proven sound only up to $O(N^{1/12})$ queries to permutations on $N$ elements. We revisit this analysis, improving the soundness bound to a tight $\Omega(N^{1/2})$. In addition to being tighter, our proof is conceptually simpler and more direct, and gives the same bound in the ideal cipher model. The main technical idea is to construct the compression isometry from a simple POVM on the naive purification, a technique which may find wider applications.

    As immediate applications, our results yield tight, concrete collision and pre-image lower bounds for the sponge hash construction underlying SHA3 and the Davies--Meyer compression function used in SHA1 and SHA2. More broadly, the improved soundness theorem provides a general-purpose tool for analyzing quantum security in settings where random permutations or ideal ciphers serve as the underlying primitive.
\end{abstract}

\section{Introduction}

Random permutations and ideal ciphers are central to the analysis of
symmetric cryptography. They model the primitives underlying constructions
such as sponge hashing and Davies--Meyer, which in turn inform the design
of widely used hash functions. In the quantum setting, an adversary may
query these primitives in superposition, in both the forward and inverse
directions. Understanding the information obtained through these queries
has been a persistent obstacle to proving security.

For random functions, Zhandry's compressed oracle
technique~\cite{Zhandry19} provides a useful way to reason about this
information. It represents the oracle's interaction with an algorithm
through a small database of input-output pairs, allowing quantum security
proofs to resemble classical lazy-sampling arguments. Extending this
approach to permutations is more difficult: injectivity correlates the
database entries, and forward and inverse queries must remain consistent.

The compressed permutation oracle of~\cite{Carolan-Comp} established
such a framework, with databases consisting of partial permutations.
Its applications include security proofs for Feistel and query lower
bounds for sponge and Davies--Meyer. However, its soundness analysis
was limited to the $N^{1/12}$ query scale for permutations on $N$
elements. This simulation error can dominate the resulting security
bounds, even when the compressed database analysis gives the expected
bound.

\paragraph{Our results.}
We give a direct analysis of the same compressed permutation oracle,
improving its soundness guarantee to the square-root query scale.
Specifically, an algorithm making $q$ forward or inverse queries
distinguishes the compressed oracle from a uniformly random permutation
with advantage at most
$O\!\left(\frac{q}{\sqrt N}\right)$.

The same bound holds for ideal ciphers, independently of the number of
keys and allowing queries in superposition over keys. Thus the
compressed simulation is sound whenever $q=o(\sqrt N)$.
Our theorem gives explicit constants, which enables proving concrete security.

The main idea is to construct the compression isometry directly from
the naive purification of a random permutation. Each partial permutation
determines a uniform superposition of its completions. These states
provide a natural candidate for decompression, but their overlaps
prevent the resulting map from being an isometry. We separate the
completion states into orthogonal levels and normalize the resulting
frame (a \emph{frame} is a collection of vectors that spans a Hilbert space but may be linearly dependent). Equivalently, compression is the coherent implementation of a
POVM on the purification. Representation theory of the symmetric group
controls this normalization and allows us to compare individual
purified and compressed queries. We develop this argument in the
technical overview and prove it in
\Cref{sec:intertwiner-soundness}.

\paragraph{Applications.}
The improved soundness bound makes the existing database method more
effective. In \Cref{sec:search-applications}, we prove an explicit-constant
fundamental lemma relating an algorithm's output to the compressed
database, and use it to establish preimage and collision resistance
of sponge and Davies--Meyer.

For a sponge with digest length $d$ and capacity $c$, the resulting
success probabilities are
\[
    O\!\left(\frac{T^2}{2^d}+\frac{T^3}{2^c}\right)
    \quad\text{and}\quad
    O\!\left(\frac{T^3}{2^d}+\frac{T^3}{2^c}\right),
\]
respectively, where $T$ includes the queries needed to obtain and
verify the output transcripts. We give explicit constants and
specialize the bounds to all four SHA-3 hash-function parameter
sets~\cite{FIPS202}, in the random permutation model.
For Davies--Meyer with $n$-bit blocks, we obtain bounds of
$O(T^2/2^n)$ for preimages and $O(T^3/2^n)$ for collisions, including
in the ideal cipher model.

\paragraph{Related work.}
Prior approaches to permutation oracles include the compressed-oracle
arguments of Rosmanis~\cite{Rosmanis22}, Unruh's framework for compressed
permutation oracles~\cite{Unruh23}, and the permutation superposition
oracles of Majenz, Malavolta, and Walter~\cite{MMW24}.
Cojocaru et al.~\cite{CHLYY25} give lifting theorems for random
permutations and ideal ciphers, with applications to sponge and
Davies--Meyer. Other quantum security analyses of these constructions
include~\cite{CGBHS18,HY18,CP24,CPZ24,ACMT25}. Our contribution is a stronger
soundness analysis of the compressed permutation oracle
of~\cite{Carolan-Comp}, allowing its database-based proofs to yield
stronger security bounds.

\section{Technical overview}

To show soundness of the compressed permutation oracle, will compare it against an ideal random permutation oracle. In particular, in the ideal experiment
we keep a register containing a uniformly random permutation and
answer queries using that permutation. We \emph{purify} the randomness of this permutation, initializing a register in the uniform superposition
\[
    \ket{S_N}=\frac{1}{\sqrt{N!}}\sum_{\varphi\in S_N}\ket{\varphi}.
\]
Keeping this register inaccessible to the algorithm gives the same
experiment as sampling a classical random permutation. In the compressed permutation oracle description,
the register, now referred to as the compressed oracle, maintains a superposition of small databases
of input-output pairs, starting from the empty database. Queries are answered using the compressed permutation oracle $\cpO$ defined in~\cite{Carolan-Comp}. At a high level, this operator answers a query by first \emph{decompressing} the database at the queried point, preparing a uniform superposition over responses if it was undefined, then answering the query by looking up in the database, followed by \emph{compressing} the database at the queried point using the inverse of decompression.

We compare these descriptions by constructing a linear map
$\widetilde V$ from the permutation register to the database register.
The map will be an \emph{isometry}, meaning that it preserves inner
products, and hence the norms of states. It maps the uniform
permutation state to the empty database and approximately satisfies
\[
    \cpO_x\widetilde V\approx\widetilde V P_x,
\]
where $P_x$ is a purified query and $\cpO_x$ is a compressed query at
input $x$. Thus answering a query and then compressing should be
close to compressing first and then answering the query.\footnote{We suppress identity operators on the answer register throughout this overview.} A hybrid argument across all queries then shows the result.

\paragraph{From databases to permutations.}
A database $f$ specifies a partial permutation: it fixes some
input-output pairs without repeating a domain or image value.
A natural way to turn it into a state of the permutation register
is to take the uniform superposition of all permutations agreeing
with those pairs. Write this \emph{completion state} as
\[
    \ket{T_f}
    =\frac{1}{\sqrt{(N-r)!}}
      \sum_{\varphi:\,\varphi\supseteq f}\ket{\varphi},
    \qquad |f|=r.
\]
This suggests the linear map
\[
    E_r\ket f=\ket{T_f}
\]
from size-$r$ databases to permutation states. However, different
databases can have common completions, so their completion states
need not be orthogonal. The map $E_r$ therefore need not preserve
norms. These completion states also appear in Rosmanis's analysis of
permutation inversion~\cite[Section~3.1]{Rosmanis22}, where their
span describes the permutation states accessible after at most
$r$ queries. We use the same states to construct an explicit
norm-preserving map from the permutation register to the
compressed database.

We first separate the information represented by different database
sizes. Let $\mathcal H_{\le r}$ be the span of completion states for
databases of size at most $r$. Define $\mathcal H_r$ to be the part
of this space orthogonal to $\mathcal H_{\le r-1}$. We call
$\mathcal H_r$ the $r$th \emph{level}: it consists of the directions
that first appear when size-$r$ databases are allowed. Let $\Pi_r$
be the orthogonal projector onto this level. Projecting the completion
map gives
\[
    \Gamma_r=\Pi_rE_r.
\]
Its outputs now lie entirely in level $r$, although they still
overlap within that level.

\paragraph{Correcting the overlaps.}
The Gram matrix
\[
    G_r=\Gamma_r\Gamma_r^\dagger
\]
describes the normalization needed to correct these overlaps.
We show that it is invertible on $\mathcal H_r$, and define
\[
    \widetilde V_r=\Gamma_r^\dagger G_r^{-1/2}.
\]
The inverse square root makes this map an isometry:
\[
    \widetilde V_r^\dagger\widetilde V_r
    =G_r^{-1/2}G_rG_r^{-1/2}=I_{\mathcal H_r}.
\]
Combining these maps over the orthogonal levels gives
$\widetilde V$.

This construction also has a measurement interpretation, which formed the original intuition behind its construction. A
\emph{positive operator-valued measure}, or POVM, is a collection
of positive operators summing to the identity; each operator
specifies the probability of one measurement outcome. On level $r$,
our outcomes are the size-$r$ databases projected to this level, i.e. the states $\Pi_r \ket{T_f}$ for $|f|=r$. If we take operators
\[
    G_r^{-1/2}\Pi_r\ket{T_f}\bra{T_f}\Pi_rG_r^{-1/2}
\]
in our POVM, then this forms the \emph{pretty good measurement} on the aforementioned states, which is, in fact, optimal for minimum-error state discrimination in our case as we have an ensemble of pure states that forms an orbit under a transitive group action~\cite{EF01,EF02}. 
The map $\widetilde V_r$ records these outcomes in a quantum register (the compressed permutation oracle),
preserving their superposition instead of measuring them.

To analyze the normalization, we use the symmetry under relabeling
the inputs and outputs of a permutation. Representation theory
decomposes the permutation space into
subspaces labeled by irreducible representations of the symmetric group. On each such subspace pieces,
$G_r$ acts by a scalar, which we calculate explicitly. In particular, we show
\[
    I\preceq G_r\preceq
    \exp\!\left(\frac{r}{N-2r+1}\right)I,
    \qquad \text{ when }r\le\frac{N-1}{2}.
\]
Here $\preceq$ denotes the ordering of self-adjoint operators:
$A\preceq B$ means that $B-A$ is positive semidefinite.
Thus the required normalization is small when $r$ is small
compared with $N$.

\paragraph{Comparing one query.}
For the query calculation, it is useful to introduce a second map,
\[
    V_r=\Gamma_r^\dagger G_r^{-1}.
\]
Unlike $\widetilde V_r$, this map is a \emph{contraction}, meaning it can decrease but not increase norms. Its advantage is that it satisfies the exact identity
\[
    E_rV_r=I_{\mathcal H_r}.
\]
Thus applying $V_r$ and then replacing each database by its
completion state recovers the original permutation state exactly.
We write $V$ for the map obtained by combining the $V_r$ over levels.

A query changes the level by at most one, in both descriptions.
Starting from level $r$, we can therefore compare the transitions
to levels $r-1$, $r$, and $r+1$ separately. Moreover, on states
produced by $V_r$, answering with the compressed oracle and then
applying the completion maps reproduces the purified query exactly.

The remaining task is to compare the database states themselves.
There are two issues. First, a compressed query may produce a
component outside the subspaces reached by $V$. Second, even within
those subspaces, the coefficients may differ from those obtained
by applying a purified query followed by $V$.

We control the first issue using cancellation identities satisfied
by states in the range of $V_r$. Fix a database with one pair
removed. If we fix the missing input and sum the coefficients over
all available outputs, the sum is zero. The analogous sum over
available inputs also vanishes. We show that, in the levels of
interest, these conditions characterize the range of $V_r$.
Quantitatively, a state whose sums are small must be close to that
range. Applying this estimate to the state after a query bounds
the component outside the range.

For the second issue, we fix the queried pair $x\mapsto y$ and
compare the normalization operators before and after that pair
is imposed. Their eigenvalues have explicit formulas, and the
relevant ratios are close to one. Together, these calculations give
\[
    \bigl\|(\cpO_xV-VP_x)\Pi_{\le t}\bigr\|
    =O\!\left(\frac{1}{\sqrt{N-t}}\right),
\]
where $\Pi_{\le t}$ projects onto the first $t+1$ levels.
The estimate holds for arbitrary states of the answer register,
including states entangled with the algorithm's workspace.

\paragraph{Accumulating the error.}
The eigenvalue bounds for $G_r$ show that $V$ and $\widetilde V$
differ by $O(t/N)$ on levels at most $t$, when $t$ is small compared
with $N$. Replacing $V$ by the isometry therefore gives a one-query
error of
\[
    O\!\left(\frac1{\sqrt N}+\frac{t+1}{N}\right).
\]
After $t$ queries, the purified permutation register is supported
on levels at most $t$. The algorithm's operations between queries
act only on its accessible registers, so they commute with the
compression map. We can consequently add the one-query errors
over the computation, obtaining final-state Euclidean distance
\[
    O\!\left(\frac q{\sqrt N}+\frac{q^2}{N}\right) = O\!\left(\frac q{\sqrt N}\right).
\]
Because $\widetilde V$ preserves norms and acts only on the
inaccessible oracle register, this also bounds the algorithm's
distinguishing advantage.

Exchanging the roles of inputs and outputs gives the same analysis
for inverse queries. For an ideal cipher, we apply an independent
compression map to the permutation associated with each key.
A query acts on only the selected key's permutation, so the same
bound holds even for superpositions over keys, without a loss
depending on the number of keys.
\section{Preliminaries}

For a finite set $X$, let $\Hh(X)=\Span\{\ket{x}:x\in X\}$ denote the corresponding Hilbert space.  All registers are finite-dimensional, and all unlabeled identity operators are understood from context.  We write $\norm{\cdot}$ for both the Euclidean norm on vectors and the operator norm on linear maps. For subspaces $W\subset V\subset \Hh$ of a finite-dimensional Hilbert space $\Hh$, we denote by $V\ominus W$ the orthocomplement of $V$ in $W$.

\subsection{Permutations}

Let $[N]=\{1,2,\ldots,N\}$ and let $S_N$ denote the symmetric group on $[N]$.  The quantum random permutation model samples a uniformly random permutation $\varphi\in S_N$ and exposes quantum oracle access to both $\varphi$ and $\varphi^{-1}$.  In the standard out of place oracle formulation these can be written as
\begin{align}
    \mathcal O_{\varphi}\ket{x,y}
        &=\ket{x,y+\varphi(x)}, &
    \mathcal O_{\varphi^{-1}}\ket{y,x}
        &=\ket{y,x+\varphi^{-1}(y)},
\end{align}
where addition in the second register is over $\mathbb Z_N$.
For our analysis, it is more convenient to use a distinguished blank symbol $\bot$.  Let $\A=\Hh({[N]\cup\{\bot\}})$ denote the extended register and, for each $z\in[N]$, let $S_z$ be the unitary on $\A$ which swaps $\ket{\bot}$ and $\ket z$ and fixes every other computational-basis state. This allows us to define the swap oracle, an alternate to the standard oracle.

\begin{definition}[Swap oracle]\label{def:swap-oracle}
    For $\varphi\in S_N$, the forward and inverse swap oracles are defined by \begin{align}
        \widetilde{\mathcal O}_{\varphi}
        &\coloneqq \sum_{x\in[N]}\proj{x}\otimes S_{\varphi(x)}, & 
        \widetilde{\mathcal O}_{\varphi^{-1}}
        &\coloneqq \sum_{x\in[N]}\proj{x}\otimes S_{\varphi^{-1}(x)}
    \end{align}
\end{definition}

We will work exclusively with swap oracles, though the addition and swap formulations are equivalent up to a constant factor in the number of queries.\footnote{One could additional define an in-place permutation oracle mapping $\ket{x}$ to $\ket{\varphi(x)}$ and the inverse. This model is also equivalent up to constant factors to the swap model.}

It will be useful to purify the choice of the random permutation.  Let
\(
    \PermO=\Hh({S_N})
\)
be the permutation-memory register, with computational basis
$\{\ket{\varphi}:\varphi\in S_N\}$. Note that this space is also known as the regular representation or the group algebra of the symmetric group, depending on which mathematical structure we equip it with, see \cref{subsec:representation-preliminaries}.
For $x,y\in[N]$, it will be useful to define a projector $Q_{x,y}$ that asserts a permutation maps $x$ to $y$, as well as purified forward and inverse queries $P_x, P_y^{-1}$ respectively.

\begin{definition} For $x, y \in [N]$, define
\begin{align}\label{eq:purified-permutation-oracle}
    Q_{x,y} &\coloneqq \sum_{\substack{\varphi\in S_N\\ \varphi(x)=y}}
      \proj{\varphi}, & P_x
    &\coloneqq \sum_{y\in[N]}S_y\otimes Q_{x,y}, & P^{-1}_y
    &\coloneqq \sum_{x\in[N]}S_x\otimes Q_{x,y},
\end{align}
where the former acts on $\PermO$ and the latter two act on $\A\otimes\PermO$. 
\label[definition]{def:perms-purified-operators}
\end{definition}

When we drop the subscript on the $P_x, P_y^{-1}$ operators these refer to the coherently controlled analogues. In particular, $P = \sum_{x \in [N]} \proj{x} \otimes P_x$, and similarly for $P^{-1}$.
\begin{remark}
If we initialize a permutation register in the uniform state $\ket{S_N} \coloneqq (N!)^{-1/2} \sum_{\varphi \in S_N} \ket{\varphi}$, and replace swap queries $\widetilde O_{\varphi}, \widetilde O_{\varphi^{-1}}$ with purified queries $P, P^{-1}$ respectively acting on the purification register, then tracing out the permutation memory after any interaction recovers the quantum random permutation model experiment with a uniformly random, classically sampled permutation.
    
\end{remark}

\subsection{Ciphers}

The ideal cipher model is the keyed analogue of the random permutation model. Letting $\mathcal K$ be a finite key space, an ideal cipher is a family
\(
    \Phi=(\varphi_k)_{k\in\mathcal K}\in S_N^{\times \mathcal K}
\)
in which the permutations $\varphi_k$ are sampled independently and uniformly.  The adversary is given quantum access to
\[
    (k,x)\longmapsto\varphi_k(x)
    \qquad\text{and}\qquad
    (k,y)\longmapsto\varphi_k^{-1}(y).
\]
A corresponding purification is obtained by taking
\[
    \PermO_{\mathrm{IC}}=\Hh(S_N^{\times\mathcal K}),
    \qquad
    \ket{\mathrm{\Phi}}_{\mathrm{IC}}
    =\frac{1}{(N!)^{|\mathcal K|/2}}
      \sum_{\Phi\in S_N^{\mathcal K}}\ket{\Phi}.
\]
For $k\in\mathcal K$ and $x,y\in[N]$, let
\[
    Q_{k,x,y}
    =\sum_{\substack{\Phi\in S_N^{\mathcal K}\\ \varphi_k(x)=y}}
      \proj{\Phi}.
\]
The corresponding purified forward and inverse oracle operators are
\[
    \PermO_{k,x}
    =\sum_{y\in[N]}S_y\otimes Q_{k,x,y},
    \qquad
    \PermO^{-1}_{k,y}
    =\sum_{x\in[N]}S_x\otimes Q_{k,x,y}.
\]
Thus a query under key $k$ acts only on the permutation indexed by $k$.  This tensor-product structure will allow the permutation analysis to extend directly to ideal ciphers.

\subsection{Compressed permutation oracles}

The compressed permutation oracle is a technique for analyzing quantum algorithms that query random, invertible permutations. The technique involves statefully simulating quantum queries through a \emph{compressed permutation oracle} operator, which updates an initially empty database with points corresponding to knowledge gained by the adversary. We recall the basic definitions of the techniques in this section, but refer to \cite{Carolan-Comp} for a more thorough treatment.

\begin{definition}
    An injective partial function
\[
    f:D\hookrightarrow[N],
    \qquad D\subseteq[N],
\]
will be called a partial permutation.  We write $\dom(f)=D$, $\im(f)=f(D)$, and $|f|=|D|$.
\label{def:partial-funcs}
\end{definition}
If $x\notin\dom(f)$ and $y\notin\im(f)$, then
$f[x\mapsto y]$ denotes the partial permutation obtained by adding the pair $(x,y)$ to $f$. Let $\I_r$ be the set of partial permutations of size $r$, and define
\[
    \Ll_r=\Hh(\I_r),
    \qquad
    \Ll=\bigoplus_{r=0}^{N}\Ll_r.
\]
We call $r$ the \emph{level} or \emph{degree} of a partial function, and extend this terminology to $\Ll_r$ and $\Ll$. The computational-basis vector indexed by $f$ is denoted $\ket f$, and the unique vector in $\Ll_0$ is denoted $\ket\varnothing$.  The compressed permutation oracle uses $\Ll$ as its database register and initializes it in $\ket\varnothing$.

Fix $x\in[N]$.  For every $h$ with $x\notin\dom(h)$, set
\[
    Y_h=[N]\setminus\im(h),
    \qquad
    h_y=h[x\mapsto y],
\]
and define
\[
    \ket{+_{x,h}}
    =\frac{1}{\sqrt{|Y_h|}}
      \sum_{y\in Y_h}\ket{h_y}.
\]
The following spaces are indexed by $h \in \I$ with $x \not\in \dom(h)$:
\[
    \Span\bigl(\ket h,\{\ket{h_y}:y\in Y_h\}\bigr),
    \qquad x\notin\dom(h).
\]
These spaces are mutually orthogonal and exhaust $\Ll$.  Let $\Comp_x$ be the unitary which, on each such space, swaps $\ket h$ and $\ket{+_{x,h}}$ and fixes the orthogonal complement of their span. 

\begin{definition}
    We can extend the assertion projectors $Q_{x,y}$ from \Cref{def:perms-purified-operators} to this space in the natural way, in particular writing
\[
    Q_{x,y}
    =\sum_{\substack{f \in \I,\\ f(x)=y}}\proj f,
    \qquad
    Q_{x,\bot}
    =\sum_{\substack{f \in \I,\\ x\notin\dom(f)}}\proj f,
\]
and their conjugates under compression
\[
    M_{x,y}=\Comp_xQ_{x,y}\Comp_x^\dagger,
    \qquad
    M_{x,\bot}=\Comp_xQ_{x,\bot}\Comp_x^\dagger.
\]
\label[definition]{def:partial-assertions}
\end{definition}
These operators form a projective measurement:
\[
    M_{x,\bot}+\sum_{y\in[N]}M_{x,y}=I_{\Ll}.
\]
With these in place, we can define the compressed permutation oracle succinctly.

\begin{definition}[Compressed permutation oracle]\label{def:compressed-permutation-oracle}
    For a fixed query point $x\in[N]$, define
    \begin{equation}\label{eq:compressed-permutation-oracle}
        \cpO_x
        =I_{\A}\otimes M_{x,\bot}
         +\sum_{y\in[N]}S_y\otimes M_{x,y}.
    \end{equation}
    The full forward oracle is obtained by controlling on the input register, $\cpO = \sum_{x \in [N]} \proj{x} \otimes \cpO_x$.
    \label[definition]{def:comp-perms}
\end{definition}

To define inverse queries, let
\(
    F:\Ll\rightarrow\Ll\) denote the flip operator such that \(
    F\ket f=\ket{f^{-1}}
\).
The compressed inverse oracle is the conjugation of the compressed oracle under flip,
\[
    \cpO^{-1}_y
    =(I_{\A}\otimes F)\cpO_y(I_{\A}\otimes F).
\]
This definition extends in the natural way to compressed \emph{cipher} oracles, with independent databases indexed by each key in $\mathcal K$. In particular, the cipher compressed database database is $\mathcal C = \bigotimes_{k \in \mathcal K}\mathcal H(\I)$, initialized as $\bigotimes_{k \in \mathcal K} \ket{\varnothing}$.

\begin{definition}[Compressed cipher oracle]\label{def:compressed-cipher-oracle}
    For a fixed query point $x\in[N]$ and key $k \in \mathcal K$, define
    \begin{equation}\label{eq:compressed-cipher-oracle}
        \ccO_{x, k}
        = (\cpO_x)_{\mathcal Ak} \otimes \left(\bigotimes_{\widetilde k \in \mathcal K \setminus \{k\}} I_{\widetilde k}\right),
    \end{equation}
    where the final subscript indicates which permutation database is acted on.
    The full forward oracle is obtained by controlling on the input and key register, $\cpO = \sum_{x \in [N], k \in \mathcal K} \proj{x, k} \otimes \cpO_{x, k}$.
    \label[definition]{def:comp-perms}
\end{definition}
An inverse cipher query is defined analogously to the inverse compressed oracle, in particular using the flip operator as 
\[
    \ccO^{-1}_{y, k}
    =(I_{\A}\otimes F^{\otimes \mathcal K})\ccO_{y, k}(I_{\A}\otimes F^{\otimes \mathcal K}).
\]
Note that the compressed permutation oracle and compressed cipher oracle can both be efficiently simulated. In the case of the compressed cipher oracle, one can store only non-empty databases, removing the need to store exponentially many compressed databases.

\subsection{Representation-theoretic preliminaries}
\label{subsec:representation-preliminaries}

In this section, we recall standard representation-theoretic facts which will be useful throughout the proof.
A partition of $n$ is a nonincreasing sequence
\[
    \lambda=(\lambda_1,\lambda_2,\ldots),
    \qquad
    \sum_i\lambda_i=n,
\]
and we write $\lambda\vdash n$.  We identify $\lambda$ with its Young diagram
\[
    \{(i,j):1\le j\le\lambda_i\}.
\]
The conjugate partition is denoted $\lambda'$.  If $\mu\subseteq\lambda$, then $\lambda/\mu$ denotes the corresponding skew diagram.  It is a \emph{horizontal strip} if it contains at most one box in each column.

For $\lambda\vdash n$, let $S^\lambda$ denote the corresponding Specht module and let
\[
    f^\lambda=\dim S^\lambda.
\]
More generally, $f^{\lambda/\mu}$ denotes the number of standard Young tableaux of skew shape $\lambda/\mu$.  The content of a box $(i,j)$ is $j-i$, and the total content of $\lambda$ is
\[
    \Ccont{\lambda}
    =\sum_{(i,j)\in\lambda}(j-i).
\]

The group $S_n\times S_n$ acts on the group algebra $\Hh(S_n)$ by
\[
    (\sigma,\tau)\ket{\varphi}
    =\ket{\sigma\circ\varphi\circ\tau^{-1}}.
\]
Restricting to the left or right action only, the associated representation is known as the regular representation. 
Under this action, the group algebra decomposes as
\begin{equation}\label{eq:regular-bimodule-decomposition}
    \Hh(S_n)
    \cong
    \bigoplus_{\lambda\vdash n}
    S^\lambda\otimes S^\lambda.
\end{equation}

We recall the representation-theoretic facts used below.

\begin{proposition}[Standard representation-theoretic facts]
\label{prop:representation-facts}
Let $\lambda\vdash n$.

\begin{enumerate}[label=(\roman*)]
    \item If $h_\lambda(i,j)$ is the hook length of $(i,j)\in\lambda$, then
    \[
        f^\lambda
        =\frac{n!}{\prod_{(i,j)\in\lambda}h_\lambda(i,j)}.
    \]

    \item Under the standard inclusion $S_m\le S_n$,
    \[
        \dim (S^\lambda)^{S_m}
        =f^{\lambda/(m)}.
    \]
    In particular, this space is nonzero if and only if $\lambda_1\ge m$.

    \item If $\alpha\vdash j$ and $m=n-j$, then the Pieri rule gives
    \[
        \operatorname{Ind}_{S_j\times S_m}^{S_n}
        \bigl(S^\alpha\boxtimes\mathbf 1\bigr)
        \cong
        \bigoplus_{\substack{\lambda\vdash n\\
        \lambda/\alpha\text{ a horizontal }m\text{-strip}}}
        S^\lambda.
    \]

    \item The transposition class sum
    \[
        T_n=\sum_{1\le a<b\le n}(a\ b)
    \]
    acts on $S^\lambda$ as multiplication by $\Ccont{\lambda}$.
\end{enumerate}
\end{proposition}

\subsection{Sponge and Davies--Meyer}
\label{subsec:sponge-davies-meyer}

The sponge and Davies--Meyer are two constructions of hash functions from permutations and ciphers, respectively. In both cases, security is analyzed against quantum algorithms
with forward and inverse access to the underlying primitive.

\paragraph{Sponge.}
The sponge construction builds a variable-input-length hash function
from a public permutation~\cite{BDPV08,BDPV11}. It is parameterized by
a \emph{rate} $r\ge1$, a \emph{capacity} $c\ge1$, and a permutation
\[
    \varphi:\{0,1\}^{r+c}\longrightarrow\{0,1\}^{r+c}.
\]
Thus the underlying permutation acts on $N=2^{r+c}$ elements. The
internal state consists of $r$ rate bits and $c$ capacity bits. Each
message block is XORed into the rate bits, after which the permutation
is applied to the entire state.

More precisely, let $m=m_1\|\cdots\|m_\ell$ be a nonempty sequence
of blocks $m_i\in\{0,1\}^r$. Starting from $s_0=0^{r+c}$, define
\[
    s_i=\varphi\bigl(s_{i-1}\oplus(m_i\|0^c)\bigr),
    \qquad 1\le i\le\ell.
\]
Let $\operatorname{pref}_r(s)$ denote the first $r$ bits of $s$.
The sponge with one output block is
\[
    \mathsf{Sponge}^{\varphi}(m)
    =\operatorname{pref}_r(s_\ell).
\]
Arbitrary bit strings are handled by first applying a fixed,
unambiguous padding rule to obtain a nonempty sequence of rate-sized
blocks. Longer outputs are obtained by repeatedly applying $\varphi$
to the final state and reading another block of rate bits. We restrict
attention here to a single output block. The standardized hash
function SHA-3 uses the sponge construction with a Keccak
permutation and specified padding and output lengths~\cite{FIPS202}.

In the random permutation model, $\varphi$ is sampled uniformly and
the adversary may query both $\varphi$ and $\varphi^{-1}$. For a fixed
target $w\in\{0,1\}^r$, the \emph{preimage problem} is to find a message
$m$ such that $\mathsf{Sponge}^{\varphi}(m)=w$. The \emph{collision
problem} is to find distinct messages $m,m'$ such that
\[
    \mathsf{Sponge}^{\varphi}(m)
    =\mathsf{Sponge}^{\varphi}(m').
\]
If padding is used, these messages are required to be valid padded
encodings. Quantum security bounds for the sponge have been studied
in~\cite{CGBHS18,CP24,ACMT25,CHLYY25,Carolan-Comp}.

\paragraph{Davies--Meyer.}
The Davies--Meyer construction builds a compression function from a
block cipher~\cite{Winternitz84}. Let
$\Phi=(\varphi_k)_{k\in\mathcal K}$ be a family of permutations on
$\{0,1\}^n$, so that $N=2^n$. The compression function is
\[
    \mathsf{DM}^{\Phi}:
    \mathcal K\times\{0,1\}^n\longrightarrow\{0,1\}^n,
    \qquad
    \mathsf{DM}^{\Phi}(k,x)=\varphi_k(x)\oplus x.
\]
Here $x$ is the chaining value and $k$ represents the message block.
The XOR with $x$ is called \emph{feed-forward}. Iterating such a
compression function gives a variable-input-length hash function.
Davies--Meyer-type constructions also allow other feed-forward
operations; SHA-1 and SHA-2 use wordwise modular addition~\cite{FIPS180-4}.
We use XOR throughout.

In the ideal cipher model, the permutations $\varphi_k$ are sampled
independently and uniformly, and the adversary has quantum access to
encryption and decryption, including superpositions over keys. For a
fixed target $w\in\{0,1\}^n$, the preimage problem is to find $(k,x)$
such that $\mathsf{DM}^{\Phi}(k,x)=w$. The collision problem is to find
distinct pairs $(k,x)$ and $(k',x')$ such that
\[
    \varphi_k(x)\oplus x
    =\varphi_{k'}(x')\oplus x'.
\]
These problems have been studied
in~\cite{HY18,CHLYY25,Carolan-Comp}. 

Fixing a single key gives the permutation-based function
\[
    \mathsf{DM}^{\varphi}(x)=\varphi(x)\oplus x.
\]
In particular, finding a preimage of $0^n$ in this specialization is
exactly the problem of finding a fixed point of $\varphi$.

Notable problems have remained open in analyzing both hash constructions, for instance tight bounds for collision resistance are not known for either in prior work.

\section{Improved soundness}
\label{sec:intertwiner-soundness}

We first treat a single random permutation.  The ideal-cipher extension is given at the end of the section and follows similarly. 

\subsection{The construction and the main theorem}

We begin with a set of states that will ease the analysis of the naive purification. For $f\in\I_r$, let
\(
    \T_f=\{\varphi\in S_N:\varphi|_{\dom(f)}=f\}
\) denote the set of (total) permutations which totalize $f$,
and define the normalized completion state
\[
    \ket{T_f}
    =\frac{1}{\sqrt{(N-r)!}}
      \sum_{\varphi\in\T_f}\ket\varphi.
\]
We can use these states to define a hierarchy of spaces determined by the span of completion states, where the $r$-th \emph{level} is spanned by completions of functions of size $r$ after removing completions of size less than $r$. We will interchangably refer to the level as the \emph{degree}. Note that the left and right symmetric group actions both preserve the aforementioned spaces, i.e., they are sub-representations of the group algebra as representation of $S_N\times S_N$. Furthermore, a quantum query to a purified oracle in the $\leq r$-th level of the hierarchy will result in a purified state lying in the $\leq r+1$-th level. We refer to this hierarchy as the \emph{permutation levels}, defined as follows.
\begin{definition}[Permutation levels]
\label{def:bounded-subspaces}
Set
\[
    \Hh_{\le r}
    =\Span\{\ket{T_f}:|f|\le r\},
    \qquad
    \Hh_r=\Hh_{\le r}\ominus\Hh_{\le r-1},
\]
where $\Hh_{\le-1}=0$, and denote the corresponding orthogonal
projectors by $\Pi_{\le r}$ and $\Pi_r$.  Also define
\begin{equation}
\label{eqn:synth-map}
    E_r:\Ll_r\longrightarrow\Hh_{\le r},
    \qquad
    E_r\ket f=\ket{T_f}.
\end{equation}
\end{definition}

The $E_r$ map represents a natural guess for a (de-)compression operator, however it turns out to have non-ideal properties. For instance, the $\ket{T_f}$ states are not orthogonal, and in fact, even fixing the size of $f$ they are not linearly independent. This means that $E_r$ is not norm-preserving, which prevents its use directly. We will instead refer to $E_r$ as the \emph{synthesis} operator, to distinguish it from (de-)compression. To define the actual intertwiner, it will be useful to write 
\begin{equation}
\label{eq:frame-objects}
    \Gamma_r=\Pi_rE_r,
    \qquad
    G_r=\Gamma_r\Gamma_r^\dagger,
    \qquad
    K_r=\ran(\Gamma_r^\dagger),
    \qquad
    \Pi_{K_r}\text{ the orthogonal projector onto }K_r.
\end{equation}
$\Gamma_r$ is a step towards the (de-)compression operator, as it removes overlap between levels: in particular, $\Gamma_r$ maps $\Ll_r$ to $\Hh_r$. Nonetheless, this operator is still not an isometry. The Gram matrix $G_r$ quantifies how far $\Gamma_r$ is from an isometry, and $K_r$ is the set of compressed database states on which de-compression is defined. $K_r$ is, up to approximation error of the compressed permutation oracle, the span of level $r$ databases which appear in the compressed oracle experiment. In this way, it plays a role similar to \emph{validity} in the standard compressed oracle theory, as defined in \cite{Zhandry19,ACMT25}. With these objects in place, we can define the compression operator $\wtV$, which we will alternately refer to as the \emph{intertwiner}. We also define a related contraction $V$ that will simplify some of the analysis.
\begin{equation}
\label{eq:two-intertwiners}
    V_r=\Gamma_r^\dagger G_r^{-1},
    \qquad
    \wtV_r=\Gamma_r^\dagger G_r^{-1/2},
    \qquad
    V=\bigoplus_{r=0}^{N-1}V_r,
    \qquad
    \wtV=\bigoplus_{r=0}^{N-1}\wtV_r.
\end{equation}

The calculation in \Cref{lem:frame-spectrum} below shows that $G_r$
is invertible on $\Hh_r$. Here $V$ is a contraction and $\wtV$ is an isometry.  The proof first establishes that (1) answering a query with the purified oracle then applying $V$, vs. (2) applying $V$ then answering a query with the compressed oracle, are close. This fact, combined with Uhlman's theorem and a standard hybrid argument, establishes our soundness bound. We call this a \textit{one-query comparison} for $V$, which we then transfers to $\wtV$. Using the contraction rather than the isometry will allow us to prove a tighter bound. Broadly, this proof technique is the same proof technique as in~\cite{Carolan-Comp}, though our construction of $V$ is more direct and inspired by~\cite{Rosmanis22}, leading to a tighter bound.

\begin{theorem}[Intertwiner and soundness]
\label{thm:intertwiner-soundness}
Let $N\ge3$.  The map $\wtV:\PermO\to\Ll$ is a degree-preserving isometry, in the sense that $\wtV$ maps $\Hh_r$ to a subspace of $\Ll_r$, satisfying
\begin{equation}
\label{eq:isometry-initial-state}
    \wtV\ket{S_N}=\ket\varnothing.
\end{equation}
For every $x\in[N]$ and every integer $t\ge0$ such that $3t+2\le N$,
\begin{align}
\label{eq:canonical-low-degree-bound}
    \left\|
      \left(
        \cpO_x(I_{\A}\otimes V)
        -(I_{\A}\otimes V)P_x
      \right)
      (I_{\A}\otimes\Pi_{\le t})
    \right\|
    &\le \frac{60}{\sqrt{N-t}},\\
\label{eq:isometric-low-degree-bound}
    \left\|
      \left(
        \cpO_x(I_{\A}\otimes\wtV)
        -(I_{\A}\otimes\wtV)P_x
      \right)
      (I_{\A}\otimes\Pi_{\le t})
    \right\|
    &\le
      \frac{60}{\sqrt{N-t}}
      +\frac{t+1}{N-2t-1}.
\end{align}
The same bounds hold for inverse queries and for the coherently controlled
forward and inverse query operators.

\end{theorem}

\begin{corollary}
Let $\mathcal A$ be an arbitrary quantum algorithm making $q$ forward and/or inverse queries, where $3q-1\le N$, and outputs a decision bit. The distinguishing advantage of the algorithm between the compressed permutation oracle from an ideal permutation is bounded by
\begin{equation}
\label{eq:main-soundness-bound}
    |\Pr_{\varphi \sim S_N}[1 \leftarrow \mathcal A^{\mathcal O_{\varphi}, \mathcal O_{\varphi^{-1}}}] - \Pr[1 \leftarrow \mathcal A^{\cpO}]|
    \le
    \frac{120q}{\sqrt{N-q+1}}
    +\frac{q(q+1)}{N-2q+1}.
\end{equation}
In particular, constant distinguishing advantage requires
$q=\Omega(\sqrt N)$.
\end{corollary}

\begin{remark}
\label{rem:isometry-versus-contraction}
The
map $V$ satisfies the helpful identity $E_rV_r=I_{\Hh_r}$, but is not an isometry.  The map $\wtV_r$ is an isometry but instead satisfies $E_r\wtV_r=G_r^{1/2}$.  Thus the phrase ``compression isometry'' refers to $\wtV$, not to $V$, though $V$ will simplify the analysis.
\end{remark}

\subsection{One query: exact reduction}

The proof is top down.  In this subsection, we reduce the theorem to three one-query bounds.  The calculations
which establish those bounds are deferred to the next
subsections.
For a family of operators $T=(T_y)_{y\in[N]}$, write
\begin{equation}
\label{eq:row-column-norms}
    \operatorname{Col}(T)
    =\left\|\sum_yT_y^\dagger T_y\right\|^{1/2},
    \qquad
    \operatorname{Row}(T)
    =\left\|\sum_yT_yT_y^\dagger\right\|^{1/2}.
\end{equation}
Fix $x\in[N]$ and abbreviate
\[
    Q_y=Q_{x,y},
    \qquad
    M_y=M_{x,y},
    \qquad
    M_\bot=M_{x,\bot}.
\]
Let $L_j$ be the orthogonal projector from $\Ll$ onto $\Ll_j$.
On the span of extensions of a $h$ at position $x$, where $x\notin\dom(h)$, put
$Y_h=[N]\setminus\im(h)$ and $s=|Y_h|$.  Directly from the definition of
$\Comp_x$,
\begin{equation}
\label{eq:subspace-projector}
    M_y|_{\Span(\ket h,\{\ket{h_z}:z\in Y_h\})}
    =\proj{u_{h,y}},
    \qquad
    \ket{u_{h,y}}
    =\ket{h_y}-\frac1s\sum_{z\in Y_h}\ket{h_z}
      +\frac1{\sqrt s}\ket h,
\end{equation}
when $y\in Y_h$, and the restriction is zero when $y\notin Y_h$. The first thing we will show is that the oracle operators change the subspace level by at most $\pm 1$, in both the compressed and purified picture.

\begin{lemma}[Bandedness]
\label{lem:bandedness}
For every $r$ and $y$,
\begin{equation}
\label{eq:true-bandedness}
    Q_y\Hh_r
    \subseteq
    \Hh_{r-1}\oplus\Hh_r\oplus\Hh_{r+1}.
\end{equation}
Likewise,
\[
    M_y\Ll_r\subseteq
    \Ll_{r-1}\oplus\Ll_r\oplus\Ll_{r+1}.
\]
Both statements hold for inverse queries.  Hence a purified computation
beginning in $\Hh_0$ is supported on $\Hh_{\le q}$ after $q$ queries.
\end{lemma}

\begin{proof}
If $f\in\I_r$ and $x\notin\dom(f)$, then
\[
    Q_y\ket{T_f}
    =
    \begin{cases}
       (N-r)^{-1/2}\ket{T_{f[x\mapsto y]}},
          &y\notin\im(f),\\
       0,&y\in\im(f).
    \end{cases}
\]
If $x\in\dom(f)$, the result is either $\ket{T_f}$ or zero.  Thus
$Q_y\Hh_{\le r}\subseteq\Hh_{\le r+1}$.  Since $Q_y$ is self-adjoint,
it maps $\Hh_r$ orthogonally to $\Hh_{\le r-2}$, proving
\eqref{eq:true-bandedness}.  The statement for $M_y$ follows from
\eqref{eq:subspace-projector}.  Inverse queries are obtained by interchanging
domains and images.
\end{proof}

We next give an alternate formula for the synthesis map $E_r$ on level $r$ which will prove easier to analyze.

\begin{lemma}[Exact synthesis]
\label{lem:exact-synthesis}
For every $a\in K_r$,
\begin{equation}
\label{eq:exact-synthesis}
    \sum_{j=r-1}^{r+1}E_jL_jM_y a=Q_yE_ra,
\end{equation}
where nonexistent levels are omitted.
\end{lemma}

\begin{proof}
Consider first the span of extensions of an $f\in\I_r$ on input $x$, with $x\notin\dom(f)$.  If $y\in Y_f$, applying
\eqref{eq:subspace-projector} to the coefficient $a_f$ and then synthesizing
gives
\[
    \frac{a_f}{\sqrt{N-r}}\ket{T_{f[x\mapsto y]}}
    -\frac{a_f}{N-r}\ket{T_f}
    +\frac{a_f}{N-r}\ket{T_f}
    =Q_y(a_f\ket{T_f}).
\]
Both sides vanish if $y\notin Y_f$.

Now consider a subspace spanned by extensions of $h\in\I_{r-1}$.  The input coefficients on
its elements are $a_{h_z}$.  By \Cref{lem:canonical-frame-properties} below,
$\sum_{z\in Y_h}a_{h_z}=0$.  Hence the coefficient against
$u_{h,y}$ is $a_{h_y}$, and the synthesized output on this subspace is
\[
    \frac{a_{h_y}}{\sqrt{|Y_h|}}\ket{T_h}
    +a_{h_y}\left(
       \ket{T_{h_y}}
       -\frac1{|Y_h|}\sum_{z\in Y_h}\ket{T_{h_z}}
     \right)
    =a_{h_y}\ket{T_{h_y}}.
\]
This is precisely the action of $Q_y$ on the synthesized input from the
subspace.  Summing over the database subspaces proves
\eqref{eq:exact-synthesis}.
\end{proof}

We will now build up to stating the key inequalities for the one-query bound. The proofs of these bounds are deferred to \Cref{subsec:joint-incidence,subsec:fixed-pair}. Before stating the inequalities, we define some additional terminology. For fixed $r$, define the linear maps
\begin{equation}
\label{eq:query-components}
    b^+_{y,r}=L_{r+1}M_yV_r,
    \qquad
    b^0_{y,r}=L_rM_yV_r.
\end{equation}
For $r\ge1$, also put
\begin{equation}
\label{eq:down-query-component}
    b^-_{y,r}=L_{r-1}M_yV_r.
\end{equation}
When $\psi\in\Hh_r$, we abbreviate
\[
    a=V_r\psi,
    \qquad
    b^+_y=b^+_{y,r}\psi,
    \qquad
    b^0_y=b^0_{y,r}\psi.
\]
For $r\ge1$, we also write $b^-_y=b^-_{y,r}\psi$.  For $r=0$, every
object carrying the superscript ${}^-$ or the subscript $r-1$ is omitted
and understood to be zero.  Put
\begin{equation}
\label{eq:orthogonal-leakage}
    z^+_{y,r}=(I-\Pi_{K_{r+1}})b^+_{y,r},
    \qquad
    z^0_{y,r}=(I-\Pi_{K_r})b^0_{y,r}.
\end{equation}
For vectors we write $z^+_y=z^+_{y,r}\psi$ and
$z^0_y=z^0_{y,r}\psi$.
The downward component belongs to $K_{r-1}$ by
\Cref{lem:downward-harmonic} below.  Define the two components which remain
inside the canonical range by
\begin{align}
\label{eq:middle-range-error}
    h^0_{y,r}\psi
    &=\Pi_{K_r}b^0_y-V_r\Pi_rQ_y\psi,\\
\label{eq:downward-range-error}
    h^-_{y,r}\psi
    &=b^-_y-V_{r-1}\Pi_{r-1}Q_y\psi.
\end{align}

The core technical inequalities we will establish can now be stated.

\begin{proposition}
\label{prop:three-query-estimates}
If $3r+2\le N$ and $m=N-r$, then
\begin{align}
\label{eq:leakage-estimate-summary}
    \operatorname{Col}(z^+_{\cdot,r}\oplus z^0_{\cdot,r})^2,
    \ \operatorname{Row}(z^+_{\cdot,r}\oplus z^0_{\cdot,r})^2
    &\le\frac{64}{m},\\
\label{eq:middle-estimate-summary}
    \operatorname{Col}(h^0_{\cdot,r})^2,
    \ \operatorname{Row}(h^0_{\cdot,r})^2
    &\le\frac{32}{m},\\
\label{eq:down-estimate-summary}
    \operatorname{Col}(h^-_{\cdot,r})^2,
    \ \operatorname{Row}(h^-_{\cdot,r})^2
    &\le\frac{24}{m}.
\end{align}
\end{proposition}

\Cref{prop:three-query-estimates} contains the cross-level terms coming from the synthesis map.  Indeed, if
$d^j_{y,r}=L_j(M_yV-VQ_y)\Pi_r$, then the exact identities are
\begin{align}
\label{eq:correct-triangular-identities}
    d^{r+1}_{y,r}\psi
      &=z^+_y,\nonumber\\
    d^r_{y,r}\psi
      &=z^0_y
        -V_r\Pi_rE_{r+1}z^+_y,\\
    d^{r-1}_{y,r}\psi
      &=-V_{r-1}\Pi_{r-1}
        \bigl(E_rz^0_y+E_{r+1}z^+_y\bigr).\nonumber
\end{align}
Equations \eqref{eq:middle-range-error} and
\eqref{eq:downward-range-error} are exactly the last terms in the second
and third lines, respectively.

To verify \eqref{eq:correct-triangular-identities}, write
$b^+_y=\Pi_{K_{r+1}}b^+_y+z^+_y$ and
$b^0_y=\Pi_{K_r}b^0_y+z^0_y$ in
\eqref{eq:exact-synthesis}.  Since $E_jK_j\subseteq\Hh_j$ by
\Cref{lem:canonical-frame-properties}, projection first onto
$\Hh_{r+1}$, then onto $\Hh_r$, and finally onto $\Hh_{r-1}$ gives
\begin{align*}
    \Pi_{r+1}Q_y\psi
       &=\Gamma_{r+1}b^+_y,\\
    \Pi_rQ_y\psi
       &=\Gamma_rb^0_y+\Pi_rE_{r+1}z^+_y,\\
    \Pi_{r-1}Q_y\psi
       &=\Gamma_{r-1}b^-_y
         +\Pi_{r-1}(E_rz^0_y+E_{r+1}z^+_y).
\end{align*}
Applying the appropriate $V_j$ and using
$V_j\Gamma_j=\Pi_{K_j}$ proves the three identities. We can now set up the one-query bound for the contraction $V$.

\begin{proposition}
\label{prop:canonical-one-query}
If $3t+2\le N$, then
\[
    \left\|
      \left(
        \cpO_x(I_{\A}\otimes V)
        -(I_{\A}\otimes V)P_x
      \right)
      (I_{\A}\otimes\Pi_{\le t})
    \right\|
    \le\frac{60}{\sqrt{N-t}}.
\]
\end{proposition}

\begin{proof}
For fixed $r$, write
\[
    D_{y,r}=(M_yV-VQ_y)\Pi_r.
\]
Its output is the sum of the four components
$z^+_y,z^0_y,h^0_{y,r}\psi,h^-_{y,r}\psi$, lying in the corresponding
orthogonal database subspaces.  The triangle inequality for both family
norms and \Cref{prop:three-query-estimates} give
\begin{equation}
\label{eq:memory-family-bound}
    \operatorname{Col}(D_{\cdot,r}),
    \ \operatorname{Row}(D_{\cdot,r})
    \le
    \frac{8+\sqrt{32}+\sqrt{24}}{\sqrt{N-r}}
    <\frac{20}{\sqrt{N-r}}.
\end{equation}

The range of $V$ is contained in $\bigoplus_rK_r$.  The harmonicity in
\Cref{lem:canonical-frame-properties} therefore implies
$M_\bot V=0$: on each subspace, $M_\bot$ projects the level-$r$ elements onto
their uniform vector, whose coefficient is the corresponding harmonic
sum, while the level-$r$ base vector is orthogonal to that projection.
Since
$M_\bot+\sum_yM_y=I$ and $\sum_yQ_y=I$, it follows that
\begin{equation}
\label{eq:D-zero-sum}
    \sum_yD_{y,r}=0.
\end{equation}
Write a vector in $\A\otimes\Hh_r$ as
\[
    \ket\Phi
    =\ket\bot\ket{\phi_\bot}
      +\sum_y\ket y\ket{\phi_y}.
\]
Using $(S_y-I)\alpha=(\alpha_\bot-\alpha_y)(\ket y-\ket\bot)$ and
\eqref{eq:D-zero-sum}, the answer-register error is
\[
    \sum_y\ket yD_{y,r}(\ket{\phi_\bot}-\ket{\phi_y})
    +\ket\bot\sum_yD_{y,r}\phi_y.
\]
The two displayed answer parts are orthogonal.  Applying the column bound
to the common input $\phi_\bot$, the individual bound
$\|D_{y,r}\|\le\operatorname{Col}(D_{\cdot,r})$ to the diagonal inputs,
and the row bound to the final sum gives
\begin{align*}
    \sum_y\|D_{y,r}(\phi_\bot-\phi_y)\|^2
      &\le 2\kappa_r^2
        \left(\|\phi_\bot\|^2+\sum_y\|\phi_y\|^2\right),\\
    \left\|\sum_yD_{y,r}\phi_y\right\|^2
      &\le\kappa_r^2\sum_y\|\phi_y\|^2,
      \qquad \kappa_r=\frac{20}{\sqrt{N-r}}.
\end{align*}
Therefore
\begin{equation}
\label{eq:answer-level-bound}
    \left\|
      \left(
        \cpO_x(I_{\A}\otimes V)
        -(I_{\A}\otimes V)P_x
      \right)
      (I_{\A}\otimes\Pi_r)
    \right\|^2
    \le\frac{1200}{N-r}.
\end{equation}

Finally decompose an input supported on $\Hh_{\le t}$ into its level
components.  By \Cref{lem:bandedness}, each output database level receives
contributions from at most three input levels.  Cauchy--Schwarz and
\eqref{eq:answer-level-bound} therefore give
\[
    \|B_x\Phi\|^2
    \le3\sum_{r=0}^t\|B_x\Pi_r\Phi\|^2
    \le\frac{3600}{N-t}\|\Phi\|^2,
\]
where
\[
    B_x=\cpO_x(I_{\A}\otimes V)-(I_{\A}\otimes V)P_x.
\]
Taking square roots proves the proposition.
\end{proof}

\begin{proof}[Proof of \Cref{thm:intertwiner-soundness}]
The frame identities in \Cref{lem:canonical-frame-properties} show that
$\wtV$ is an isometry and that it maps $\ket{S_N}$ to
$\ket\varnothing$.  Equation \eqref{eq:canonical-low-degree-bound} is
\Cref{prop:canonical-one-query}.

Let $\Delta=\wtV-V$.  By \Cref{lem:frame-spectrum}, for
$s\le(N-1)/2$,
\begin{equation}
\label{eq:polar-dual-distance}
    \|\Delta\Pi_s\|
    \le
    \frac{s}{2(N-2s+1)}.
\end{equation}
Both $\cpO_x$ and $P_x$ are unitary, and $P_x$ maps
$\Hh_{\le t}$ into $\Hh_{\le t+1}$.  Thus
\begin{align*}
    &\left\|
      \left(
        \cpO_x(I_{\A}\otimes\wtV)
        -(I_{\A}\otimes\wtV)P_x
      \right)
      (I_{\A}\otimes\Pi_{\le t})
    \right\|\\
    &\qquad\le
      \frac{60}{\sqrt{N-t}}
      +\frac{t}{2(N-2t+1)}
      +\frac{t+1}{2(N-2t-1)}\\
    &\qquad\le
      \frac{60}{\sqrt{N-t}}
      +\frac{t+1}{N-2t-1},
\end{align*}
which is \eqref{eq:isometric-low-degree-bound}.

Let $\mathsf{Inv}\ket\varphi=\ket{\varphi^{-1}}$ on $\PermO$.  The flip
$F$ on $\Ll$ and $\mathsf{Inv}$ preserve levels, and the definitions give
\[
    FV=V\mathsf{Inv},
    \qquad
    F\wtV=\wtV\mathsf{Inv}.
\]
Conjugating proves the inverse-query estimates.  Coherently controlled
queries are direct sums over the query point, so they obey the same norm
bounds.

For the hybrid argument, let $\delta_j$ be the Euclidean distance between
the compressed state immediately after query $j$ and the image under
$\wtV$ of the corresponding purified state.  The initial distance is zero
by \eqref{eq:isometry-initial-state}.  Inter-query unitaries act only on
accessible registers and commute with $\wtV$.  Immediately before query
$j+1$, \Cref{lem:bandedness} places the purified memory in
$\Hh_{\le j}$.  Hence
\[
    \delta_{j+1}
    \le\delta_j
       +\frac{60}{\sqrt{N-j}}
       +\frac{j+1}{N-2j-1}.
\]
Summing and using $3q-1\le N$ gives
\begin{equation}
\label{eq:isometric-hybrid-distance}
    \delta_q
    \le
    \frac{60q}{\sqrt{N-q+1}}
    +\frac{q(q+1)}{2(N-2q+1)}.
\end{equation}
Both states in this comparison are normalized because $\wtV$ is an
isometry.  For any effect on the accessible registers, the difference of
its expectation values is at most twice their Euclidean distance.  Applying
this to \eqref{eq:isometric-hybrid-distance} proves
\eqref{eq:main-soundness-bound}.
\end{proof}

\subsection{The POVM isometry}

In this subsection, we use the symmetry of the Gram matrix $G$ and representation theory to establish key properties. The primary property is that $G$ is close to the identity, and is established in \Cref{lem:frame-spectrum} through a direct calculation of the spectrum. This in turn shows that $\wtV$ is close to an isometry. 
We use the notation $\lambda\vdash n$ for a partition of $n$, identified
with its Young diagram, and write $\lambda'$ for the conjugate partition.
The relation $\alpha\nearrow\lambda$ means that $\lambda$ is obtained
from $\alpha$ by adding one box.  A skew diagram is a \emph{horizontal
strip} if it contains at most one box in each column.  Let $S^\lambda$ be
the Specht module indexed by $\lambda$, let
$f^\lambda=\dim S^\lambda$, and, for a skew shape $\lambda/\nu$, let
$f^{\lambda/\nu}$ be its number of standard Young tableaux.  We use the
standard branching, Pieri, Schur-orthogonality, and hook-length formulas
for these modules.  Under the left-right action
\[
    (\sigma,\tau)\ket\varphi=\ket{\sigma\varphi\tau^{-1}},
\]
we have
\begin{equation}
\label{eq:regular-decomposition}
    \PermO\cong
    \bigoplus_{\lambda\vdash N}S^\lambda\otimes S^\lambda.
\end{equation}

\begin{lemma}
\label{lem:frame-spectrum}
Let $0\le r\le N-1$ and let $m=N-r$.  The space $\Hh_r$ is the direct sum
of the blocks in \eqref{eq:regular-decomposition} with
$\lambda_1=m$.  If $\lambda=(m,\mu)$ on such a block, then $G_r$ acts by
\begin{equation}
\label{eq:frame-eigenvalue}
    g_{m,\mu}
    =\frac1{m!}
      \prod_{c=1}^{m}\bigl(m-c+1+\mu'_c\bigr)
    =\prod_{c=1}^{m}
      \left(1+\frac{\mu'_c}{m-c+1}\right).
\end{equation}
In particular,
\begin{equation}
\label{eq:frame-extremes}
    \lambda_{\min}(G_r)=\frac{N}{N-r},
    \qquad
    \|G_r^{-1}\|=\frac{N-r}{N}\le1.
\end{equation}
If $r\le(N-1)/2$, then
\begin{equation}
\label{eq:frame-upper-bound}
    \|G_r\|
    \le
    \exp\left(\frac{r}{N-2r+1}\right).
\end{equation}
Consequently, \eqref{eq:polar-dual-distance} holds.
\end{lemma}

\begin{proof}
For $\pi,\sigma\in S_N$,
\[
    \bra\pi E_rE_r^\dagger\ket\sigma
    =\frac1{m!}
      \binom{|\{u:\pi(u)=\sigma(u)\}|}{r}.
\]
Thus $E_rE_r^\dagger$ is central.  The class function
$r!\binom{\operatorname{fix}(\tau)}r$ is the character of the
permutation representation on ordered injective $r$-tuples.  Schur
orthogonality and Young branching therefore show that its scalar on
$S^\lambda\otimes S^\lambda$ is
\begin{equation}
\label{eq:general-frame-scalar}
    \binom Nr\frac{f^{\lambda/(m)}}{f^\lambda}.
\end{equation}
This is nonzero exactly when $\lambda_1\ge m$.  Hence
$\Hh_{\le r}$ contains precisely those blocks, while $\Hh_r$ contains the
blocks with $\lambda_1=m$.

For $\lambda=(m,\mu)$, hook-length cancellation in
\eqref{eq:general-frame-scalar} gives \eqref{eq:frame-eigenvalue}.  Since
$\sum_c\mu'_c=r$,
\[
    g_{m,\mu}
    \ge1+\sum_c\frac{\mu'_c}{m}
    =\frac Nm,
\]
with equality for $\mu=(1^r)$, proving
\eqref{eq:frame-extremes}.  If $\mu'_c\ne0$, then
$c\le\mu_1\le r$, and hence
\[
    \log g_{m,\mu}
    \le\sum_c\frac{\mu'_c}{m-c+1}
    \le\frac{r}{m-r+1}.
\]
This proves \eqref{eq:frame-upper-bound}.

On a block with frame eigenvalue $g$, the singular values of $V_r$ and
$\wtV_r$ are $g^{-1/2}$ and $1$, respectively.  Therefore
\[
    \|(\wtV_r-V_r)\Pi_r\|
    =\max_g(1-g^{-1/2})
    \le\frac12\max_g\log g
    \le\frac{r}{2(N-2r+1)},
\]
which is \eqref{eq:polar-dual-distance}.
\end{proof}

\begin{lemma}
\label{lem:canonical-frame-properties}
For every $r$,
\begin{align}
\label{eq:canonical-polar-identities}
    \Gamma_rV_r&=I_{\Hh_r}, &
    V_r^\dagger V_r&=G_r^{-1}, &
    \|V_r\|&\le1,\\
    \wtV_r^\dagger\wtV_r&=I_{\Hh_r}, &
    \wtV_r\wtV_r^\dagger&=\Pi_{K_r}, &
    E_r\wtV_r&=G_r^{1/2}.
\end{align}
Moreover, for every $a\in K_r$,
\begin{equation}
\label{eq:harmonicity-and-synthesis}
    \sum_{b\notin\im(h)}a_{h[u\mapsto b]}=0,
    \qquad
    \sum_{u\notin\dom(h)}a_{h[u\mapsto b]}=0,
    \qquad
    E_ra=\Gamma_ra\in\Hh_r.
\end{equation}
Here the first identity holds for every $h\in\I_{r-1}$ and
$u\notin\dom(h)$, and the second for every $h\in\I_{r-1}$ and
$b\notin\im(h)$; for $r=0$ they are vacuous.  In particular,
\begin{equation}
\label{eq:canonical-right-inverse}
    E_rV_r=I_{\Hh_r},
    \qquad
    \Pi_{K_r}=V_r\Gamma_r.
\end{equation}
Finally, the operators
\begin{equation}
\label{eq:frame-povm}
    G_r^{-1/2}
    \Pi_r\ket{T_f}\!\bra{T_f}\Pi_r
    G_r^{-1/2},
    \qquad f\in\I_r,
\end{equation}
form a POVM on $\Hh_r$, whose coherent measurement is $\wtV_r$.
\end{lemma}

\begin{proof}
The identities in \eqref{eq:canonical-polar-identities} follow directly
from \eqref{eq:frame-objects}, \eqref{eq:two-intertwiners}, and
\Cref{lem:frame-spectrum}.  Their sum in \eqref{eq:frame-povm} is
$G_r^{-1/2}G_rG_r^{-1/2}=I$.

Let $a=\Gamma_r^\dagger\psi\in K_r$.  If
$h\in\I_{r-1}$ and $u\notin\dom(h)$, then
\begin{align*}
    \sum_{b\notin\im(h)}a_{h[u\mapsto b]}
    &=\left\langle
      \sum_{b\notin\im(h)}\Pi_r\ket{T_{h[u\mapsto b]}},
      \psi
      \right\rangle\\
    &=\sqrt{N-r+1}\,\ip{\Pi_r\ket{T_h}}{\psi}=0.
\end{align*}
The corresponding sum over unused domains vanishes in the same way.  This
proves the first two identities in \eqref{eq:harmonicity-and-synthesis}.

For the last identity, use centrality rather than a termwise overlap
calculation.  By \Cref{lem:frame-spectrum}, $E_rE_r^\dagger$ preserves each
block in \eqref{eq:regular-decomposition}, and $\Hh_r$ is a sum of whole
blocks.  Hence
\[
    E_ra=E_rE_r^\dagger\psi\in\Hh_r,
\]
so $E_ra=\Pi_rE_ra=\Gamma_ra$.  Applying this to $a=V_r\psi$ proves the
first identity in \eqref{eq:canonical-right-inverse}; the second follows
algebraically from $V_r\Gamma_r=\Gamma_r^\dagger G_r^{-1}\Gamma_r$.
\end{proof}

\subsection{Proof of \Cref{eq:leakage-estimate-summary}}
\label{subsec:joint-incidence}

In this subsection, we will prove \Cref{eq:leakage-estimate-summary}.
For $1\le j\le N$, define the maps
\begin{align}
    \Idown_j:\Ll_j&\longrightarrow
    \bigoplus_{\substack{h\in\I_{j-1}\\b\notin\im(h)}}
       \mathbb C\ket{h,b}, &
    (\Idown_ja)_{h,b}
    &=\frac1{\sqrt{N-j+1}}
      \sum_{u\notin\dom(h)}a_{h[u\mapsto b]},
    \label{eq:image-incidence}\\
    \Ddown_j:\Ll_j&\longrightarrow
    \bigoplus_{\substack{h\in\I_{j-1}\\u\notin\dom(h)}}
       \mathbb C\ket{h,u}, &
    (\Ddown_ja)_{h,u}
    &=\frac1{\sqrt{N-j+1}}
      \sum_{b\notin\im(h)}a_{h[u\mapsto b]}.
    \label{eq:domain-incidence}
\end{align}
Write $\partial_j=(\Ddown_j,\Idown_j)$ and set
$\Ddown_0=\Idown_0=0$.
Thus the two zero-sum relations in
\eqref{eq:harmonicity-and-synthesis} say precisely that
$\partial_ja=0$ for $a\in K_j$.

\begin{lemma}
\label{lem:joint-incidence-gap}
Let $1\le j\le N/2$ and let $n=N-j$.  Then
\begin{equation}
\label{eq:joint-incidence-gap}
    \partial_j^\dagger\partial_j
    \succeq
    \frac{n-j+1}{n+1}(I-\Pi_{K_j}).
\end{equation}
Consequently, if $j\le(N+1)/3$, then
\begin{equation}
\label{eq:incidence-coercivity}
    \|(I-\Pi_{K_j})b\|^2
    \le2\bigl(\|\Ddown_jb\|^2+\|\Idown_jb\|^2\bigr).
\end{equation}
\end{lemma}

\begin{proof}
Let $[N]^{(j)}$ be the injective ordered $j$-tuples.  The matching space
$\Ll_j$ is the diagonal-$S_j$ invariant subspace of
$\Hh([N]^{(j)})\otimes\Hh([N]^{(j)})$.  Young branching and the Pieri rule
give
\[
    \Hh([N]^{(j)})
    \cong
    \bigoplus_{\substack{
       \alpha\vdash j,\ \lambda\vdash N\\
       \lambda/\alpha\text{ a horizontal }n\text{-strip}}}
       S^\alpha\otimes S^\lambda.
\]
Taking diagonal-$S_j$ invariants in the tensor square leaves the summands
$S^\lambda\otimes S^\mu$ indexed by a common $\alpha\vdash j$ such that
both $\lambda/\alpha$ and $\mu/\alpha$ are horizontal $n$-strips.

Let $A_{\mathrm d}$ replace one domain entry of an ordered matching by an
unused domain value, and let $A_{\mathrm i}$ do the same on the image side.
Directly from \eqref{eq:image-incidence} and
\eqref{eq:domain-incidence},
\begin{equation}
\label{eq:incidence-adjacency}
    \partial_j^\dagger\partial_j
    =\frac{2jI+A_{\mathrm d}+A_{\mathrm i}}{n+1}.
\end{equation}
For a partition $\rho$, put
$\Ccont{\rho}=\sum_{(a,b)\in\rho}(b-a)$.  The transposition class sum acts
on $S^\rho$ by $\Ccont{\rho}$.  Splitting this class sum into
transpositions internal to the first $j$ entries, internal to the final
$n$ entries, and crossing between them shows that the eigenvalue of
\eqref{eq:incidence-adjacency} on the summand indexed by
$(\alpha,\lambda,\mu)$ is
\begin{equation}
\label{eq:incidence-eigenvalue}
    \frac{e_\alpha(\lambda)+e_\alpha(\mu)}{n+1},
    \qquad
    e_\alpha(\lambda)
    =\Ccont{\lambda}-\Ccont{\alpha}-\binom n2+j.
\end{equation}

A horizontal $n$-strip is obtained by adding boxes at the bottoms of $n$
distinct columns.  Since $j\le n$, the smallest content is obtained from
the first $n$ columns; call the resulting partition
$\alpha[N]=(n,\alpha_1,\alpha_2,\ldots)$.  Its content excess is
\[
    \sum_{c=1}^{n}(c-1-\alpha'_c)=\binom n2-j,
\]
so $e_\alpha(\alpha[N])=0$.  Any other choice of $n$ columns omits a
column at most $n$ and includes a column at least $n+1$.  Legality forces
an omitted column to have positive height in $\alpha$: otherwise the new
column heights would fail to be nonincreasing.  Thus the omitted column
is at most $j$, and the exchange raises the content by at least
$n-j+1$.  Repeating the exchange shows that every nonzero eigenvalue in
\eqref{eq:incidence-eigenvalue} is at least $(n-j+1)/(n+1)$.

The zero eigenspace consists of the blocks
$S^{\alpha[N]}\otimes S^{\alpha[N]}$.  These are precisely the blocks with
first row $n$, so \Cref{lem:frame-spectrum} and
\eqref{eq:harmonicity-and-synthesis} identify this kernel with $K_j$.
This proves \eqref{eq:joint-incidence-gap}; the final assertion follows
because $(n-j+1)/(n+1)\ge1/2$ when $j\le(N+1)/3$.
\end{proof}

Fix $a\in K_r$, let $m=N-r$ and $M=m+1$, and in this subsection write
$b^+_y=L_{r+1}M_ya$, $b^0_y=L_rM_ya$, and
$b^-_y=L_{r-1}M_ya$ when $r\ge1$.  If
$x\notin\dom(h)$, also put
$X_h=[N]\setminus\dom(h)$ and $Y_h=[N]\setminus\im(h)$.  The subspace
formula gives
\begin{align}
    (b^+_y)_{f_z}
    &=\frac{a_f}{\sqrt m}
      \one_{\{x\notin\dom(f),\ y,z\in Y_f\}}
      \left(\one_{\{z=y\}}-\frac1m\right),
    \label{eq:upward-coefficients}\\
    (b^0_y)_f
    &=\frac{a_f}{m}
      \one_{\{x\notin\dom(f),\ y\in Y_f\}},
    \label{eq:middle-absent-coefficients}\\
    (b^0_y)_{h_z}
    &=a_{h_y}
      \one_{\{x\notin\dom(h),\ y,z\in Y_h\}}
      \left(\one_{\{z=y\}}-\frac1M\right),
    \label{eq:middle-present-coefficients}\\
    (b^-_y)_h
    &=\frac{a_{h_y}}{\sqrt M}
      \one_{\{x\notin\dom(h),\ y\in Y_h\}}.
    \label{eq:downward-coefficients}
\end{align}

\begin{lemma}
\label{lem:downward-harmonic}
If $3r+2\le N$, then $b^-_y\in K_{r-1}$ for every $a\in K_r$ and $y$.
The statement is vacuous for $r=0$.
\end{lemma}

\begin{proof}
For $r=1$, this follows from $K_0=\Ll_0$.  For $r\ge2$, substituting
\eqref{eq:downward-coefficients} into either incidence map at level $r-1$
produces, up to the common factor $M^{-1/2}$, one of the two zero-sum
relations in \eqref{eq:harmonicity-and-synthesis}, with the edge
$x\mapsto y$ held fixed.  Thus $\partial_{r-1}b^-_y=0$.  Since
$r-1\le N/2$, \Cref{lem:joint-incidence-gap} gives
$\ker\partial_{r-1}=K_{r-1}$.
\end{proof}

\begin{lemma}
\label{lem:incidence-calculation}
Assume $r\le m$.  Stack
\[
    \Theta_{y,r}a
    =\bigl(
       \Idown_{r+1}b^+_y,
       \Ddown_{r+1}b^+_y,
       \Idown_rb^0_y,
       \Ddown_rb^0_y
      \bigr).
\]
Then, as operators on $K_r$,
\begin{equation}
\label{eq:incidence-family-bounds}
    \operatorname{Col}(\Theta_{\cdot,r})^2\le\frac8m,
    \qquad
    \operatorname{Row}(\Theta_{\cdot,r})^2\le\frac{22}{m}.
\end{equation}
\end{lemma}

\begin{proof}
For $r=0$, the middle incidence terms vanish and the downward term does
not exist.  Here $m=N$, and
\[
    \sum_y\|\Idown_1b^+_y(a)\|^2
       =\frac{N-1}{N^2}\|a\|^2,
    \qquad
    \left\|\sum_y\Idown_1b^+_y(a^{(y)})\right\|^2
       \le\frac1{N^2}\sum_y\|a^{(y)}\|^2,
\]
while $\Ddown_1b^+_y=0$.  This proves the lemma at $r=0$.  Assume from
now on that $r\ge1$.  Substitution of
\eqref{eq:upward-coefficients}--\eqref{eq:middle-present-coefficients} into
the incidence maps, followed by the two harmonic zero sums, gives the
following nonzero rows.  For the upward image incidence,
\begin{align}
    (\Idown_{r+1}b^+_y)_{p,b}
    &=\frac{a_p}{m}
      \one_{\{x\notin\dom(p),\ y,b\in Y_p\}}
      \left(\one_{\{b=y\}}-\frac1m\right),
    \label{eq:image-up-a}\\
    (\Idown_{r+1}b^+_y)_{h_z,b}
    &=-\frac{a_{h_b}}m
      \one_{\{y\in Y_h\setminus\{b\},\ b\ne z\}}
      \left(\one_{\{z=y\}}-\frac1m\right).
    \label{eq:image-up-b}
\end{align}
For upward domain incidence, the rows based at $p$ with
$x\notin\dom(p)$ vanish.  If $p=h_z$,
$u\in X_h\setminus\{x\}$, and $A_b=a_{h[u\mapsto b]}$, then
\begin{equation}
\label{eq:domain-up}
    (\Ddown_{r+1}b^+_y)_{h_z,u}
    =-\frac{\one_{\{y\in Y_h\}}}{m}
      \left(\one_{\{z=y\}}-\frac1m\right)
      \left(A_z+\one_{\{y\ne z\}}A_y\right),
    \qquad
    \sum_{b\in Y_h}A_b=0.
\end{equation}
For the middle level,
\begin{align}
    (\Idown_rb^0_y)_{p,b}
    &=\frac{\one_{\{y\in Y_p\}}}{\sqrt M}
      \left[
        a_{p_y}\left(\one_{\{b=y\}}-\frac1M\right)
        -\one_{\{b\ne y\}}\frac{a_{p_b}}m
      \right],
    \label{eq:image-middle}\\
    (\Ddown_rb^0_y)_{p,u}
    &=-\frac{
       \one_{\{x\notin\dom(p),\ u\ne x,\ y\in Y_p\}}
      }{m\sqrt M}
      a_{p[u\mapsto y]},
    \label{eq:domain-middle-a}\\
    (\Ddown_rb^0_y)_{h_z,u}
    &=\frac{
       \one_{\{y\in Y_h\setminus\{z\},\
                    u\in X_h\setminus\{x\}\}}
      }{M^{3/2}}
      a_{h[u\mapsto z,x\mapsto y]}.
    \label{eq:domain-middle-b}
\end{align}
Image-incidence rows based at $p$ with $x\in\dom(p)$ vanish.

We first estimate the column square function.  The two expressions in
\eqref{eq:image-up-a}--\eqref{eq:image-up-b} are multiples of
$I_m-m^{-1}J_m$, where $J_m$ is the $m\times m$ all-ones matrix; this is
an orthogonal projector of squared Frobenius norm $m-1$.  Counting the
possible removed edge gives
\begin{equation}
\label{eq:image-up-column-bound}
    \sum_y\|\Idown_{r+1}b^+_y\|^2
    \le\frac1m\|a\|^2.
\end{equation}
For a zero-sum vector $(A_b)_{b\in Y_h}$, a direct expansion of
\eqref{eq:domain-up} gives
\[
    \sum_{y,z}|(\Ddown_{r+1}b^+_y)_{h_z,u}|^2
    =\frac{m^2-1}{m^4}\sum_b|A_b|^2.
\]
Every coefficient of $a$ occurs for at most $r$ choices of the removed
edge, whence
\begin{equation}
\label{eq:domain-up-column-bound}
    \sum_y\|\Ddown_{r+1}b^+_y\|^2
    \le\frac r{m^2}\|a\|^2.
\end{equation}
Using $|s+t|^2\le2|s|^2+2|t|^2$ in
\eqref{eq:image-middle}, and directly counting the occurrences of each
coefficient in \eqref{eq:domain-middle-a}--\eqref{eq:domain-middle-b},
gives
\begin{align}
\label{eq:middle-incidence-column-bounds}
    \sum_y\|\Idown_rb^0_y\|^2
       &\le\frac4M\|a\|^2,\\
    \sum_y\|\Ddown_rb^0_y\|^2
       &\le
       \left(\frac r{m^2M}+\frac{r-1}{M^3}\right)\|a\|^2.
\end{align}
Since $r\le m$ and $M=m+1$, the right-hand sides of
\eqref{eq:image-up-column-bound}--
\eqref{eq:middle-incidence-column-bounds} sum to at most
$8m^{-1}\|a\|^2$.

For the row square function, take an arbitrary family of inputs
$a^{(y)}\in K_r$ and write $b^\bullet_y(a^{(y)})$ for the preceding
linear expressions evaluated at $a^{(y)}$.  Sum the corresponding
incidence vectors over $y$.
Equations \eqref{eq:image-up-a}--\eqref{eq:image-up-b} again apply
$I_m-m^{-1}J_m$, now to the vector of inputs indexed by $y$, and yield
\begin{equation}
\label{eq:image-up-row-bound}
    \left\|\sum_y\Idown_{r+1}b^+_y(a^{(y)})\right\|^2
    \le\frac1{m^2}\sum_y\|a^{(y)}\|^2.
\end{equation}
Expanding \eqref{eq:domain-up} gives three terms,
\[
    -\frac{m-1}{m^2}a^{(z)}_{h[u\mapsto z]},
    \qquad
    \frac1{m^2}\sum_{y\ne z}a^{(y)}_{h[u\mapsto z]},
    \qquad
    \frac1{m^2}\sum_{y\ne z}a^{(y)}_{h[u\mapsto y]}.
\]
Cauchy--Schwarz, followed by counting the rows in which a fixed
coefficient occurs, gives
\begin{equation}
\label{eq:domain-up-row-bound}
    \left\|\sum_y\Ddown_{r+1}b^+_y(a^{(y)})\right\|^2
    \le\frac9{m^2}\sum_y\|a^{(y)}\|^2.
\end{equation}
Indeed, the second displayed term counts a fixed coefficient at most $r$
times, and the third at most $m$ times.  Splitting
\eqref{eq:image-middle} into its diagonal term and two off-diagonal sums
and applying Cauchy--Schwarz to each gives
\begin{equation}
\label{eq:image-middle-row-bound}
    \left\|\sum_y\Idown_rb^0_y(a^{(y)})\right\|^2
    \le\frac9M\sum_y\|a^{(y)}\|^2.
\end{equation}
In \eqref{eq:domain-middle-a}, each coefficient is counted once after the
Cauchy--Schwarz sum over $y$; in \eqref{eq:domain-middle-b}, it is counted
at most $r-1$ times.  Consequently,
\begin{align}
\label{eq:domain-middle-row-bound}
    \left\|\sum_y\Ddown_rb^0_y(a^{(y)})\right\|^2
    &\le
      \left(\frac1{m^2}+\frac{r-1}{M^2}\right)
      \sum_y\|a^{(y)}\|^2\nonumber\\
    &\le\frac2m\sum_y\|a^{(y)}\|^2.
\end{align}
The right-hand sides of
\eqref{eq:image-up-row-bound}--\eqref{eq:domain-middle-row-bound} sum to
at most $22m^{-1}\sum_y\|a^{(y)}\|^2$.  Taking the supremum over the
common input for the column bound and over the tuple $(a^{(y)})_y$ for
the row bound proves \eqref{eq:incidence-family-bounds}.
\end{proof}

\begin{proof}[Proof of \eqref{eq:leakage-estimate-summary}]
For each positive $j\in\{r,r+1\}$, the hypothesis $3r+2\le N$ implies
$j\le(N+1)/3$.  By \Cref{lem:joint-incidence-gap}, the restriction of
$\partial_j$ to $K_j^\perp$ has a left inverse of norm at most $\sqrt2$.
Since $\partial_j\Pi_{K_j}=0$, applying this left inverse to the two
relevant components of \Cref{lem:incidence-calculation} gives column
constant at most $2\cdot8=16$ and row constant at most
$2\cdot22=44$.  The uniform value $64$ proves
\eqref{eq:leakage-estimate-summary}.  The $r=0$ case is obtained by
observing that the middle leakage is zero and omitting the nonexistent
downward term.  Finally, composing the maps on $K_r$ furnished by
\Cref{lem:incidence-calculation} with the contraction $V_r$ preserves
both family bounds, so the conclusion applies to the operators in
\eqref{eq:orthogonal-leakage}.
\end{proof}

\subsection{Proof of \Cref{eq:middle-estimate-summary,eq:down-estimate-summary}}
\label{subsec:fixed-pair}

In this subsection, we will prove \Cref{eq:middle-estimate-summary,eq:down-estimate-summary}. On $\Ll_r$, define the diagonal selectors
\begin{align}
\label{eq:fixed-pair-selectors}
    R^{00}_{x,y;r}
    &=\sum_{\substack{f\in\I_r\\
             x\notin\dom(f),\ y\notin\im(f)}}\proj f,
    &
    R^{11}_{x,y;r}
    &=\sum_{\substack{f\in\I_r\\f(x)=y}}\proj f.
\end{align}
Put
\begin{equation}
\label{eq:fixed-pair-frames}
    A_{y,r}=\Pi_rQ_y\Pi_r,
    \qquad
    H^{00}_{y,r}=\frac1m\Gamma_rR^{00}_{x,y;r}\Gamma_r^\dagger,
    \qquad
    H^{11}_{y,r}=\Gamma_rR^{11}_{x,y;r}\Gamma_r^\dagger.
\end{equation}

\begin{lemma}
\label{lem:exact-middle-frame}
For $a,c\in K_r$,
\begin{equation}
\label{eq:middle-quadratic-form}
    \ip{c}{L_rM_yL_ra}
    =\left\langle c,
      \left(\frac1mR^{00}_{x,y;r}+R^{11}_{x,y;r}\right)a
      \right\rangle.
\end{equation}
Consequently,
\begin{equation}
\label{eq:middle-error-frame-formula}
    h^0_{y,r}
    =V_r\left(
       H^{00}_{y,r}+H^{11}_{y,r}-A_{y,r}G_r
      \right)G_r^{-1}
    =V_r\left(
       H^{00}_{y,r}G_r^{-1}
       +H^{11}_{y,r}G_r^{-1}-A_{y,r}
      \right).
\end{equation}
\end{lemma}

\begin{proof}
On an $x$-absent level-$r$ subspace, the level-$r$ coordinate of
$u_{f,y}$ is $m^{-1/2}\ket f$, giving
$m^{-1}\overline{c_f}a_f$.  On a subspace based at
$h\in\I_{r-1}$, the level-$r$ part of $u_{h,y}$ is
$\ket{h_y}-M^{-1}\sum_z\ket{h_z}$.  The two harmonic zero sums in
\eqref{eq:harmonicity-and-synthesis} reduce its inner products with $a$
and $c$ to $a_{h_y}$ and $c_{h_y}$.  This proves
\eqref{eq:middle-quadratic-form} by polarization.

Since $\Pi_{K_r}=V_r\Gamma_r$ and
$V_r\psi=\Gamma_r^\dagger G_r^{-1}\psi$, the $K_r$ part of the middle
compressed transition is
\[
    \Pi_{K_r}L_rM_yL_rV_r
    =V_r(H^{00}_{y,r}+H^{11}_{y,r})G_r^{-1}.
\]
Subtracting $V_rA_{y,r}$ proves
\eqref{eq:middle-error-frame-formula}.
\end{proof}

\begin{lemma}
\label{lem:middle-fixed-pair-comparison}
If $3r+2\le N$, then
\[
    \operatorname{Col}(h^0_{\cdot,r})^2,
    \ \operatorname{Row}(h^0_{\cdot,r})^2
    \le\frac{32}{m}.
\]
\end{lemma}

\begin{proof}
If $r=0$, then $G_0=I$, $R^{11}_{x,y;0}=0$, and
$H^{00}_{y,0}=A_{y,0}=I/N$, so $h^0_{y,0}=0$.  Assume henceforth that
$r\ge1$.

It is enough to bound the operator in parentheses in the second expression
of \eqref{eq:middle-error-frame-formula}, because $V_r$ is a contraction.
We first compare the true assertion with the fixed-edge frame.

By left-right symmetry, fix $x=y=N$ and write $H=S_{N-1}$.  Under the
$H\times H$ stabilizer action, Young branching is multiplicity free:
\[
    S^\lambda\!\downarrow_H
    =\bigoplus_{\alpha\nearrow\lambda}S^\alpha.
\]
In normalized matrix-coefficient bases, let $\rho^\lambda$ denote a
unitary matrix realization of $S^\lambda$.  Schur orthogonality gives
\begin{equation}
\label{eq:fixed-coset-seed}
    \frac{\sqrt{f^\lambda f^\nu}}{N!}
    \sum_{h\in H}
       \overline{\rho^\lambda(h)_{ab}}
       \rho^\nu(h)_{cd}
    =\frac{\sqrt{f^\lambda f^\nu}}{Nf^\alpha}
      \delta_{ac}\delta_{bd}
\end{equation}
on a common diagonal $S^\alpha\otimes S^\alpha$ branch, and zero on
unequal branches.  Thus $A_{y,r}$ has, on the multiplicity space of this
branch, the rank-one seed
\[
    v^\alpha(v^\alpha)^\dagger,
    \qquad
    v^\alpha_\lambda
    =\sqrt{\frac{f^\lambda}{Nf^\alpha}},
\]
where only parents $\lambda$ with $\lambda_1=m$ are retained.

Deleting the fixed edge $x\mapsto y$ identifies $H^{11}_{y,r}$ with the
size-$(r-1)$ completion frame on $S_{N-1}$.  Indeed, if
$\pi(x)=\sigma(x)=y$, then
\[
    \bra\pi E_rR^{11}_{x,y;r}E_r^\dagger\ket\sigma
    =\frac1{m!}
      \binom{|\{u\ne x:\pi(u)=\sigma(u)\}|}{r-1},
\]
and the matrix element is zero otherwise.  Hence its seed is
$h_\alpha v^\alpha(v^\alpha)^\dagger$, where
\begin{equation}
\label{eq:smaller-frame-scalar}
    h_\alpha
    =\binom{N-1}{r-1}
       \frac{f^{\alpha/(m)}}{f^\alpha}.
\end{equation}
Every child of a partition $\lambda$ with $\lambda_1=m$ has first row
$m$ or $m-1$.

If $\alpha_1=m$, write $\alpha=(m,\eta)$.  Every relevant parent
$\lambda=(m,\mu)$ is obtained by adding one lower box to $\eta$, say in
column $d$.  The product formula \eqref{eq:frame-eigenvalue} gives
\begin{equation}
\label{eq:fixed-edge-frame-ratio}
    1-\frac{h_\alpha}{g_{m,\mu}}
    =\frac1{m-d+1+\eta'_d+1}
    \le\frac1{m-r+2}.
\end{equation}
Since $\|v^\alpha\|\le1$, the norm of the corresponding seed of
$A_{y,r}-H^{11}_{y,r}G_r^{-1}$ is at most the final quantity in
\eqref{eq:fixed-edge-frame-ratio}.  More explicitly, this seed is
$v^\alpha(w^\alpha)^\dagger$, where
\[
    w^\alpha_\lambda
      =v^\alpha_\lambda
       \left(1-\frac{h_\alpha}{g_\lambda}\right),
\]
where $g_\lambda$ is the eigenvalue of $G_r$ on the parent $\lambda$.

If $\alpha_1=m-1$, then $h_\alpha=0$ and there is a unique relevant parent
$\lambda=(m,\mu)$.  Hook cancellation gives
\begin{align}
\label{eq:first-row-branch-weight}
    \frac{f^\lambda}{Nf^\alpha}
    &=\prod_{c=1}^{m-1}
      \frac{m-c+\mu'_c}{m-c+1+\mu'_c}\\
    &\le
      \prod_{c=\mu_1+1}^{m-1}\frac{m-c}{m-c+1}
      =\frac1{m-\mu_1}
      \le\frac1{m-r}.
\end{align}
The hypothesis gives $m\ge2r+2$, so in both cases
\begin{equation}
\label{eq:fixed-edge-per-outcome}
    \|A_{y,r}-H^{11}_{y,r}G_r^{-1}\|\le\frac2m.
\end{equation}
As $N=m+r\le3m/2$, the crude row and column estimates obtained by summing
the squares of \eqref{eq:fixed-edge-per-outcome} are both at most $6/m$.

It remains to bound the free-frame term
$C_y=H^{00}_{y,r}G_r^{-1}$.  Positivity gives
\[
    0\preceq H^{00}_{y,r}\preceq\frac1mG_r,
    \qquad
    \sum_yH^{00}_{y,r}\preceq G_r,
\]
where the second relation uses
\[
    \sum_yR^{00}_{x,y;r}
      =m\sum_{\substack{f\in\I_r\\x\notin\dom(f)}}\proj f
      \preceq mI_{\Ll_r}.
\]
Since $\|G_r\|\le e$ in the present range, we have
$(H^{00}_{y,r})^2\preceq(e/m)H^{00}_{y,r}$.  Thus, without assuming that
$G_r$ and $H^{00}_{y,r}$ commute,
\begin{align*}
    \sum_yC_y^\dagger C_y
      &=G_r^{-1}\left(\sum_y(H^{00}_{y,r})^2\right)G_r^{-1}
        \preceq\frac emG_r^{-1},\\
    \sum_yC_yC_y^\dagger
      &=\sum_yH^{00}_{y,r}G_r^{-2}H^{00}_{y,r}
        \preceq\sum_y(H^{00}_{y,r})^2
        \preceq\frac emG_r.
\end{align*}
Here $G_r^{-1}\preceq I$, and hence $G_r^{-2}\preceq I$, was used.
Taking norms yields
$\operatorname{Col}(C)^2\le e/m$ and
$\operatorname{Row}(C)^2\le e^2/m$.
Combining these bounds with
$\|X+Y\|^2\le2\|X\|^2+2\|Y\|^2$ gives
\[
    \operatorname{Col}(h^0_{\cdot,r})^2
      \le\frac{2(6+e)}m<\frac{18}m,
    \qquad
    \operatorname{Row}(h^0_{\cdot,r})^2
      \le\frac{2(6+e^2)}m<\frac{27}m.
\]
The common bound $32/m$ proves the lemma and
\eqref{eq:middle-estimate-summary}.
\end{proof}

\begin{lemma}
\label{lem:downward-frame-comparison}
If $3r+2\le N$, then
\[
    \operatorname{Col}(h^-_{\cdot,r})^2,
    \ \operatorname{Row}(h^-_{\cdot,r})^2
    \le\frac{24}{m}.
\]
\end{lemma}

\begin{proof}
The case $r=0$ is vacuous, so assume $r\ge1$.
Let $q_{y,r}=\Pi_{r-1}Q_y\Pi_r$.  If
$a=V_r\psi=\Gamma_r^\dagger G_r^{-1}\psi$, then
\eqref{eq:downward-coefficients} gives, coefficient by coefficient,
\[
    b^-_y
    =\Gamma_{r-1}^\dagger q_{y,r}G_r^{-1}\psi.
\]
Indeed, for $h\in\I_{r-1}$,
\[
    \left(\Gamma_{r-1}^\dagger
       q_{y,r}G_r^{-1}\psi\right)_h
    =\braket{T_h}{Q_yG_r^{-1}\psi}
    =\frac1{\sqrt{m+1}}
       \braket{T_{h[x\mapsto y]}}{G_r^{-1}\psi},
\]
with both sides zero when the extension is not legal.  In the first
equality, the projector $\Pi_{r-1}$ may be removed because
\Cref{lem:bandedness} gives
$Q_yG_r^{-1}\psi\perp\Hh_{\le r-2}$.  Therefore
\begin{equation}
\label{eq:downward-frame-residual}
    h^-_{y,r}
    =V_{r-1}
      \left(G_{r-1}q_{y,r}G_r^{-1}-q_{y,r}\right).
\end{equation}

It remains to estimate the operator in parentheses.  In the common
$S_{N-1}$ branching basis, a nonzero channel goes from
$\lambda=(m,\mu)$ to
$\nu=(m+1,\mu-\square_d)$, where $\square_d$ is a removable box in column
$d$.  Let $q_c=m-c+1$.  Applying
\eqref{eq:frame-eigenvalue} at the two adjacent levels gives
\begin{equation}
\label{eq:downward-frame-ratio}
    \rho_{\lambda,\nu}
    :=\frac{g_{\nu}}{g_{\lambda}}
    =\frac{q_d}{q_d+1}
      \prod_{\substack{1\le c\le m\\c\ne d}}
      \left(
        1-\frac{\mu'_c}
        {(q_c+1)(q_c+\mu'_c)}
      \right).
\end{equation}
All factors lie in $[0,1]$, and
\begin{align}
\label{eq:downward-ratio-bound}
    0\le1-\rho_{\lambda,\nu}
    &\le
      \frac1{m-r+2}
      +\frac{r}{(m-r+1)(m-r+2)}
      \le\frac4m.
\end{align}
For completeness, on a common $S_{N-1}$ branch $\alpha$, the assertion
$q_{y,r}$ has a rank-one seed $u^\alpha(v^\alpha)^\dagger$: the output
weight vector $u^\alpha$ is supported on the unique relevant $\nu$, and
the input vector has coordinates
\[
    v^\alpha_\lambda=\sqrt{\frac{f^\lambda}{Nf^\alpha}}.
\]
Both weight vectors have norm at most one by the branching rule.  The
residual in \eqref{eq:downward-frame-residual} has seed
\[
    u^\alpha
    \bigl((\delta^\alpha_\lambda v^\alpha_\lambda)_\lambda\bigr)^\dagger,
    \qquad
    \delta^\alpha_\lambda=\rho_{\lambda,\nu}-1.
\]
It therefore has norm at most
$\max_\lambda|\delta^\alpha_\lambda|\le4/m$ for each $y$.
Since $N\le3m/2$, both its row and column family norms squared are at most
$16N/m^2\le24/m$.  Finally $V_{r-1}$ is a contraction, so
\eqref{eq:downward-frame-residual} proves the lemma and
\eqref{eq:down-estimate-summary}.
\end{proof}

\subsection{Ideal cipher extension}
We can readily extend this framework to the ideal cipher model as follows.

\begin{corollary}[Ideal-cipher extension]
\label{cor:ideal-cipher-extension}
For the ideal cipher of \Cref{def:compressed-cipher-oracle}, give each key $k\in\mathcal K$
an independent compressed database $\Ll^{(k)}$ and let a query under key
$k$ act on the $k$th database factor.  On
\[
    \PermO_{\mathrm{IC}}
      =\bigotimes_{k\in\mathcal K}\PermO^{(k)},
    \qquad
    \Ll_{\mathrm{IC}}
      =\bigotimes_{k\in\mathcal K}\Ll^{(k)},
\]
the tensor-product map
\[
    \wtV_{\mathrm{IC}}
      =\bigotimes_{k\in\mathcal K}\wtV^{(k)}
\]
is an isometry.  With degree defined as the sum of the database degrees
over all keys, the isometric one-query bound and the soundness bound of
\Cref{thm:intertwiner-soundness} hold without any dependence on
$|\mathcal K|$.
\end{corollary}

\begin{proof}
A query controlled by $k$ is a direct sum over keys and acts nontrivially
on only one tensor factor.  On total degree at most $t$, the degree of that
factor is at most $t$, so \eqref{eq:isometric-low-degree-bound} applies.
The other factors carry isometries and do not change the norm.  A query
raises total degree by at most one, after which the hybrid proof of
\Cref{thm:intertwiner-soundness} applies verbatim.
\end{proof}

\begin{corollary}
\label{cor:sqrt-N-lower-bound}
Any algorithm which distinguishes the compressed and purified experiments
with constant advantage makes $\Omega(\sqrt N)$ forward or inverse
queries.  The same statement holds in the ideal-cipher model.
\end{corollary}

\begin{proof}
If $q=o(\sqrt N)$, then both terms on the right-hand side of
\eqref{eq:main-soundness-bound} are $o(1)$.  Therefore constant advantage
is impossible.  The ideal-cipher statement follows from
\Cref{cor:ideal-cipher-extension}.
\end{proof}

\section{Applications}
\label{sec:search-applications}

We apply the soundness theorem to search problems in the random permutation
and ideal cipher models. We first give an explicit-constant version of the
fundamental lemma of~\cite{Carolan-Comp}, relating an algorithm's output to
the contents of the compressed database. We then bound the probability of
creating a satisfying database and apply the resulting bound to sponge
and Davies--Meyer.

All query counts in this section refer to the swap oracles of
\Cref{def:swap-oracle}. One standard XOR query can be simulated by two swap
queries: compute the oracle value into a blank register, XOR it into the
answer register, and uncompute the auxiliary register. Consequently, bounds
for an algorithm making $q$ standard XOR queries follow by replacing its
query count below by $2q$.

For $3T-1\le N$, write
\begin{equation}
\label{eq:search-soundness-error}
 \delta_N(T)
 =\frac{60T}{\sqrt{N-T+1}}
  +\frac{T(T+1)}{2(N-2T+1)}.
\end{equation}
By \eqref{eq:isometric-hybrid-distance}, the final compressed state of a
$T$-query computation is within Euclidean distance $\delta_N(T)$ of the
image of the corresponding purified state under $\wtV$. The same statement
holds for ideal ciphers. We use this state-distance bound, rather than the
distinguishing-advantage bound, so that the error is added to the square
root of the success probability.

\subsection{The fundamental lemma}
\label{subsec:explicit-fundamental}

A database property is a set $\mathcal B$ of partial permutations. We call
it increasing if $f\in\mathcal B$ and $f\subseteq g$ imply
$g\in\mathcal B$. Let
\[
 B=\sum_{f\in\mathcal B}\ket f\bra f
\]
be its projector. A list of input-output pairs witnesses $\mathcal B$ if
it is consistent with a partial permutation belonging to $\mathcal B$.
Repeated pairs are allowed and impose no additional condition.

\begin{lemma}
\label{lem:explicit-fundamental}
Suppose an algorithm makes at most $q$ queries and outputs a list of at
most $\ell$ input-output pairs. Let $p$ be the probability that the list
witnesses an increasing property $\mathcal B$ and every pair agrees with
the real oracle, on measuring the purification. Let $p_{\mathrm{db}}$ be the probability that the final
database belongs to $\mathcal B$ when the same algorithm uses the
compressed oracle, on measuring the compressed database. If $3(q+\ell)-1\le N$, then
\begin{equation}
\label{eq:explicit-fundamental}
 \sqrt p
 \le \sqrt{p_{\mathrm{db}}}
    +\ell\sqrt{\frac{2}{N-q}}
    +\delta_N(q+\ell).
\end{equation}
The same bound holds for an ideal cipher, with lists of triples
$(k,x,y)$ and increasing properties of the family of keyed databases.
\end{lemma}

\begin{proof}
We first compare checking a pair against the compressed oracle with
checking whether it occurs in the database. Here $Q_{x,y}$ denotes the
partial-permutation projector from \Cref{def:partial-assertions}, whereas
$M_{x,y}=\Comp_xQ_{x,y}\Comp_x^\dagger$ is its compressed conjugate.
Fix an $x$-absent base $h$ and let $s=N-|h|$. If $y\notin\im(h)$,
these two projectors restrict to
\[
 \ket{h_y}\bra{h_y}
 \quad\text{and}\quad
 \ket{u_{h,y}}\bra{u_{h,y}},
 \qquad
 \ket{u_{h,y}}
 =\ket{h_y}-\frac1s\sum_{z\notin\im(h)}\ket{h_z}
       +\frac1{\sqrt s}\ket h.
\]
Both vectors are normalized and
$\langle h_y\mid u_{h,y}\rangle=1-1/s$. The norm of the difference
of two rank-one orthogonal projectors is the sine of the angle between
their defining vectors. Thus the norm on this subspace is exactly
\[
 \sqrt{1-(1-1/s)^2}
 =\sqrt{\frac2s-\frac1{s^2}}.
\]
If $y\in\im(h)$, both projectors vanish on the subspace. Every subspace
meeting a database of size at most $q$ has $|h|\le q$, so, writing
$L_{\le q}$ for the projector onto $\bigoplus_{j=0}^q\Ll_j$,
\begin{equation}
\label{eq:single-pair-recording}
 \bigl\|(M_{x,y}-Q_{x,y})L_{\le q}\bigr\|
 \le \sqrt{\frac2{N-q}}.
\end{equation}

For a fixed output list, abbreviate its projectors by $M_i,Q_i$.
Checking the listed pairs sequentially with fresh blank answer registers
has acceptance amplitude $M_\ell\cdots M_1$. Indeed, conditioning a
swap query on the answer $y$ applies $M_{x,y}$ to the database.
The $M_i$ need not commute; their product is a contraction, and is not
being treated as an orthogonal projector. In contrast,
$Q_\ell\cdots Q_1$ is the projector asserting that the entire list
occurs in the database. The telescoping identity
\[
 M_\ell\cdots M_1-Q_\ell\cdots Q_1
 =\sum_{i=1}^\ell
 M_\ell\cdots M_{i+1}(M_i-Q_i)Q_{i-1}\cdots Q_1
\]
and \eqref{eq:single-pair-recording} give
\begin{equation}
\label{eq:list-recording}
 \bigl\|(M_\ell\cdots M_1-Q_\ell\cdots Q_1)L_{\le q}\bigr\|
 \le\ell\sqrt{\frac2{N-q}}.
\end{equation}
There is no increase in the degree cutoff on the right: each $Q_i$ is
diagonal and preserves database size. Output-controlled lists obey the
same estimate by direct sums. Lists that do not witness $\mathcal B$
are rejected in both checks.

Append this verifier to the algorithm. It makes at most $q+\ell$
queries. Projection onto its accessible acceptance flag, together with
\eqref{eq:search-soundness-error}, shows that the square roots of its
real and compressed acceptance probabilities differ by at most
$\delta_N(q+\ell)$. Equation~\eqref{eq:list-recording} bounds the latter
by $\sqrt{p_{\mathrm{db}}}+\ell\sqrt{2/(N-q)}$, because a database
containing a witnessing list belongs to $\mathcal B$. This proves
\eqref{eq:explicit-fundamental}.

For an ideal cipher, each subspace belongs to one key. Its number of
available values is $N-|h_k|\ge N-q$, where $h_k$ is that key's base
database. The proof is otherwise unchanged, using the ideal-cipher
state-distance bound from \Cref{cor:ideal-cipher-extension}.
\end{proof}

\subsection{A general search bound}
\label{subsec:explicit-search}

We recall the sparsity framework from \cite{Carolan-Comp,Zhandry19,CFHL21} for proving compressed oracle search lower bounds.
For an increasing property $\mathcal B$ with
$\varnothing\notin\mathcal B$, define its two-sided sparsity by
\begin{align}
\label{eq:explicit-sparsity}
 s_t=\max_{\substack{|h|\le t\\h\notin\mathcal B}}\Bigl\{
 \max\Bigl\{&
 \max_{x\notin\dom(h)}
 \bigl|\{y\notin\im(h):h[x\mapsto y]\in\mathcal B\}\bigr|,
 \nonumber\\[-2mm]
 &\max_{y\notin\im(h)}
 \bigl|\{x\notin\dom(h):h[x\mapsto y]\in\mathcal B\}\bigr|
 \Bigr\}\Bigr\}.
\end{align}
For a keyed database, take an additional maximum over the key being
queried, and use total database size in the condition $|h|\le t$.

\begin{lemma}[Creating a satisfying database]
\label{lem:explicit-progress}
After $q<N$ compressed queries, the probability that the database
belongs to $\mathcal B$ satisfies
\begin{equation}
\label{eq:explicit-progress}
 \sqrt{p_{\mathrm{db}}}
 \le 2\sqrt2\sum_{t=0}^{q-1}\sqrt{\frac{s_t}{N-t}}.
\end{equation}
This holds for forward and inverse queries, including coherent choices
of input, direction, and key.
\end{lemma}

\begin{proof}
Fix a forward-query subspace based at $h$, with $x\notin\dom(h)$, and
let $s=N-|h|$. If $h\in\mathcal B$, every extension is in
$\mathcal B$, and the query cannot cross between $B$ and $I-B$ on
this subspace. Otherwise, let $b$ of its $s$ extensions belong to
$\mathcal B$. The matrix of $B\Comp_x(I-B)$ has one column for
$\ket h$ and $s-b$ columns for the nonsatisfying extensions. Each of
its $b$ rows is
\[
 \left(\frac1{\sqrt s},
       -\frac1s,\ldots,-\frac1s\right).
\]
It consequently has norm
\[
 \sqrt{\frac bs+\frac{b(s-b)}{s^2}}
 =\sqrt{\frac{2b}s-\frac{b^2}{s^2}}
 \le\sqrt{\frac{2b}s}.
\]
Since $\Comp_x$ is self-adjoint, its commutator with $B$ has the
same norm as this off-diagonal block.

Write the query as
$\cpO_x=(I_{\A}\otimes\Comp_x)U_x(I_{\A}\otimes\Comp_x)$,
where $U_x$ answers the query from the uncompressed database and
acts as the identity when $x$ is absent. The operator $U_x$ commutes
with $B$. Therefore, on the whole subspace,
\[
 \|[B,\cpO_x]\|\le2\|[B,\Comp_x]\|
 \le2\sqrt{2b/s},
\]
with identity operators on the answer register suppressed. The query
preserves this subspace. A subspace meeting an input of size at most $t$
has $|h|\le t$, and hence $b\le s_t$ and $s\ge N-t$. Taking
direct sums gives
\[
 \|B\cpO_x(I-B)L_{\le t}\|
 \le2\sqrt{\frac{2s_t}{N-t}}.
\]
Inverse queries follow by exchanging domain and image. Coherent
controls are direct sums, so do not increase the bound.

Let $a_t$ be the norm of the satisfying component after $t$ queries.
Intermediate operations commute with $B$, and a compressed query
increases database size by at most one. Unitarity and the preceding
estimate yield
\[
 a_{t+1}\le a_t+2\sqrt{\frac{2s_t}{N-t}}.
\]
Since $a_0=0$, summing proves the claim.
\end{proof}

Combining the two lemmas gives the following explicit version of the
general search bound in~\cite{Carolan-Comp}.
\begin{theorem}[Explicit search bound]
\label{thm:explicit-search}
Under the hypotheses of \Cref{lem:explicit-fundamental}, with
$\varnothing\notin\mathcal B$,
\begin{equation}
\label{eq:explicit-search}
 p\le\min\left\{1,
 \left(
 2\sqrt2\sum_{t=0}^{q-1}\sqrt{\frac{s_t}{N-t}}
 +\ell\sqrt{\frac2{N-q}}
 +\delta_N(q+\ell)
 \right)^2\right\}.
\end{equation}
\end{theorem}

To state the applications compactly, let
\begin{equation}
\label{eq:application-errors}
 a_N(T)=\sqrt{\frac N{N-T}},
 \qquad
 e_N(T)=T\sqrt{\frac2{N-T}}+\delta_N(T).
\end{equation}
Here and below an upper bound on a probability may always be truncated
at one. If $N\ge16$ and $1\le T\le\sqrt N/4$, then
\begin{equation}
\label{eq:application-error-simplification}
 a_N(T)\le\sqrt2,
 \qquad e_N(T)\le\frac{88T}{\sqrt N}.
\end{equation}
Indeed, the two terms with square-root denominators in $e_N(T)$ sum
to at most $(2+60\sqrt2)T/\sqrt N$. Its remaining term is at most
$T(T+1)/N\le T/(2\sqrt N)$.

\subsection{Sponge}
\label{subsec:explicit-sponge}

Allow the digest length to be any $1\le d\le r$. Thus, following
absorption, the sponge returns the first $d$ bits of its final state.
This includes the single-output-block construction of
\Cref{subsec:sponge-davies-meyer} and all four SHA-3 hash functions.
Write $b=r+c$ and $N=2^b$. For a state $x$, let $x_{\mathrm C}$
denote its final $c$ bits, and let $x_{[d]}$ denote its first $d$ bits.

We give the combinatorial argument explicitly, including the effect of
inverse queries. Associate with a partial permutation $f$ the directed
multigraph on $\{0,1\}^c$ having one edge
\[
 x_{\mathrm C}\longrightarrow y_{\mathrm C}
 \quad\text{for each }(x,y)\in f.
\]
Edges retain their full input-output labels. Call a vertex reachable
if there is a directed path to it from $0^c$, allowing the empty path.
An edge is reachable if its source is reachable. Every sponge
computation whose transcript occurs in $f$ follows reachable edges.

Call $f$ internally bad if a reachable edge ends at $0^c$, or two
distinct reachable edges have the same endpoint. Otherwise its
reachable graph is a rooted tree with edges directed away from the
root. In particular, every reachable vertex has a unique path from
$0^c$. To see this, a reachable cycle either enters the root or gives
a vertex two distinct incoming reachable edges; the same alternative
holds for two different paths to a vertex.

For a fixed target $w\in\{0,1\}^d$, let
$\mathcal B_{\mathrm{pre}}$ consist of internally bad databases and
databases containing a reachable edge $(x,y)$ with $y_{[d]}=w$.
Let $\mathcal B_{\mathrm{col}}$ consist of internally bad databases
and databases containing two distinct reachable edges $(x,y),(x',y')$
with $y_{[d]}=y'_{[d]}$. Both properties are increasing and exclude
the empty database.

\begin{lemma}[Sponge sparsity]
\label{lem:explicit-sponge-sparsity}
The two-sided sparsities of these properties satisfy
\begin{align}
\label{eq:sponge-pre-sparsity}
 s_t(\mathcal B_{\mathrm{pre}})
 &\le 2^{b-d}+(2t+1)2^r,\\
\label{eq:sponge-col-sparsity}
 s_t(\mathcal B_{\mathrm{col}})
 &\le t2^{b-d}+(2t+1)2^r.
\end{align}
\end{lemma}

\begin{proof}
Fix a nonsatisfying database $h$ of size at most $t$. Its reachable
vertex set has size at most $t+1$. Consider adding a forward pair
$x\mapsto y$. If $x_{\mathrm C}$ is not reachable, no new vertex or
edge becomes reachable, so the property remains false.

Suppose $x_{\mathrm C}$ is reachable. Exclude values of
$y_{\mathrm C}$ that are already reachable or occur as the source
of an edge of $h$. There are at most $(t+1)+t=2t+1$ such capacity
values, accounting for at most $(2t+1)2^r$ choices of $y$.
For every remaining choice, the new reachable edge ends at a new leaf
and makes no old edge reachable. It cannot create an internal bad
event. For preimages, it can satisfy the property only if
$y_{[d]}=w$, allowing at most $2^{b-d}$ further choices. For
collisions, its digest must equal that of an existing reachable edge,
allowing at most $t2^{b-d}$ further choices.

For an inverse query, $y$ is fixed and $x$ varies. Unless
$x_{\mathrm C}$ is already reachable, adding the edge again leaves
the reachable graph unchanged. There are at most $(t+1)2^r$ choices
with reachable $x_{\mathrm C}$. This is bounded by either displayed
right-hand side. Restricting to legal extensions only decreases all
these counts.
\end{proof}

A successful preimage transcript witnesses $\mathcal B_{\mathrm{pre}}$.
Two successful collision transcripts witness $\mathcal B_{\mathrm{col}}$:
if their database is not internally bad, the two distinct padded
messages must end on different edges. Indeed, the terminal edge
determines the unique rooted path, and that path determines every
message block by
$m_i=(x_i)_{\mathrm R}\oplus(s_{i-1})_{\mathrm R}$.
Thus identical terminal edges would imply identical padded messages.
Here the padding rule is assumed injective, as is the SHA-3 encoding.

We keep track of the cost of obtaining these transcripts. Suppose an
algorithm makes $q$ queries and outputs one candidate message, or two
candidate messages for collision finding. Let $L\ge1$ bound the total
number of padded blocks in its output. Evaluating these messages and
recording their input-output pairs uses at most $L$ additional forward
swap queries. Applying \Cref{thm:explicit-search} to this augmented
algorithm uses at most another $L$ queries for the list verification
in the fundamental lemma. Accordingly, set
\begin{equation}
\label{eq:sponge-total-query-budget}
 T=q+2L.
\end{equation}
This accounting applies to ordinary message outputs; it does not assume
that the adversary also supplies the intermediate states.

\begin{theorem}[Explicit sponge security]
\label{thm:explicit-sponge}
Let $1\le d\le r$, $N=2^{r+c}$, and $T=q+2L$. For a uniformly
random permutation and $3T-1\le N$, the fixed-target preimage and
collision success probabilities satisfy
\begin{align}
\label{eq:sponge-pre-exact}
 p_{\mathrm{pre}}
 &\le\left[
 a_N(T)\left(
 \frac{2\sqrt2\,T}{2^{d/2}}+
 \frac{4T^{3/2}}{2^{c/2}}\right)+e_N(T)
 \right]^2,\\
\label{eq:sponge-col-exact}
 p_{\mathrm{col}}
 &\le\left[
 a_N(T)\left(
 \frac{4\sqrt2\,T^{3/2}}{3\cdot2^{d/2}}+
 \frac{4T^{3/2}}{2^{c/2}}\right)+e_N(T)
 \right]^2.
\end{align}
In particular, if $1\le T\le2^{(r+c)/2}/4$, then
\begin{align}
\label{eq:sponge-pre-simple}
 p_{\mathrm{pre}}
 &\le \frac{64T^2}{2^d}
      +\frac{128T^3}{2^c}
      +\frac{32768T^2}{2^{r+c}},\\
\label{eq:sponge-col-simple}
 p_{\mathrm{col}}
 &\le \frac{32T^3}{2^d}
      +\frac{128T^3}{2^c}
      +\frac{32768T^2}{2^{r+c}}.
\end{align}
The preimage target is fixed independently of the permutation.
\end{theorem}

\begin{proof}
Let $q'=q+L$. The augmented algorithm makes at most $q'$ queries
and outputs a transcript of at most $L$ pairs. Use the appropriate
property from \Cref{lem:explicit-sponge-sparsity} in
\Cref{thm:explicit-search}. The relevant success event for the original
algorithm is contained in the event that its transcript witnesses this
property. The sum for preimages is bounded using
\[
 \sum_{t=0}^{q'-1}\sqrt{s_t}
 \le q'2^{(b-d)/2}
      +2^{r/2}\sum_{t=0}^{q'-1}\sqrt{2t+1}
 \le q'2^{(b-d)/2}+\sqrt2\,(q')^{3/2}2^{r/2}.
\]
For collisions, use instead
\[
 \sum_{t=0}^{q'-1}\sqrt t
 \le\int_0^{q'}\sqrt u\,du=\frac23(q')^{3/2}.
\]
Replacing $q'$ and $L$ by the larger value $T$ and every $N-t$ by
$N-T$ gives \eqref{eq:sponge-pre-exact} and
\eqref{eq:sponge-col-exact}.

For the simpler estimates, apply
\eqref{eq:application-error-simplification} and
$(u+v+w)^2\le3(u^2+v^2+w^2)$. The preimage coefficients are at most
$48,96,23232$, and the collision coefficients at most
$64/3,96,23232$. Rounding them upwards gives the displayed constants.
\end{proof}

\subsection{The SHA-3 parameters}
\label{subsec:explicit-sha3}

The four SHA-3 hash functions use a $1600$-bit permutation, digest
length $d\in\{224,256,384,512\}$, capacity $c=2d$, and rate
$r=1600-2d$~\cite{FIPS202}. In each case $d\le r$, so the digest is
obtained from the first squeezing block. The encoding appends the
suffix $01$ and then the multi-rate padding $10^*1$. For a message
of $m$ bits, its padded block count is
\begin{equation}
\label{eq:sha3-padded-length}
 L(m)=\left\lceil\frac{m+4}{r}\right\rceil.
\end{equation}
For collisions, use the sum of the two padded block counts, or a
public upper bound on that sum, in \eqref{eq:sponge-total-query-budget}.

Modeling the underlying permutation as uniform, the preceding theorem
gives, for $1\le T\le2^{798}$,
\begin{align}
\label{eq:sha3-pre-full}
 p_{\mathrm{pre}}
 &\le\frac{64T^2}{2^d}
       +\frac{128T^3}{2^{2d}}
       +\frac{32768T^2}{2^{1600}},\\
\label{eq:sha3-col-full}
 p_{\mathrm{col}}
 &\le\frac{32T^3}{2^d}
       +\frac{128T^3}{2^{2d}}
       +\frac{32768T^2}{2^{1600}}.
\end{align}
In particular, all four parameter sets satisfy the simpler bounds
\begin{equation}
\label{eq:sha3-concrete}
 \boxed{\quad
 p_{\mathrm{pre}}\le\min\{1,2^{7-d}T^2\},
 \qquad
 p_{\mathrm{col}}\le\min\{1,2^{6-d}T^3\}.
 \quad}
\end{equation}
To verify the first bound, it suffices to consider $T\le2^{d/2}$.
After dividing \eqref{eq:sha3-pre-full} by $T^2/2^d$, its coefficient
is at most
$64+128\cdot2^{-d/2}+32768\cdot2^{d-1600}<128$.
For collisions the corresponding coefficient is at most
$32+128\cdot2^{-d}+32768\cdot2^{d-1600}<64$, using $T\ge1$.
For $T>2^{d/2}$, both claims follow from the trivial bound of one.
Thus \eqref{eq:sha3-concrete} holds for every $T\ge1$.

\begin{center}
\begin{tabular}{lrrcc}
\hline
Function & Rate $r$ & Capacity $c$
 & Preimage probability & Collision probability\\
\hline
SHA3-224 & 1152 & 448  & $T^2/2^{217}$ & $T^3/2^{218}$\\
SHA3-256 & 1088 & 512  & $T^2/2^{249}$ & $T^3/2^{250}$\\
SHA3-384 & 832  & 768  & $T^2/2^{377}$ & $T^3/2^{378}$\\
SHA3-512 & 576  & 1024 & $T^2/2^{505}$ & $T^3/2^{506}$\\
\hline
\end{tabular}
\end{center}

For example, success probability at least $1/2$ requires
\[
 T\ge2^{d/2-4}\quad\text{for preimages},
 \qquad
 T\ge2^{(d-7)/3}\quad\text{for collisions}.
\]
The resulting explicit thresholds are
\begin{center}
\begin{tabular}{lcc}
\hline
Function & Preimages & Collisions\\
\hline
SHA3-224 & $2^{108}$ & $2^{217/3}\approx2^{72.33}$\\
SHA3-256 & $2^{124}$ & $2^{83}$\\
SHA3-384 & $2^{188}$ & $2^{377/3}\approx2^{125.67}$\\
SHA3-512 & $2^{252}$ & $2^{505/3}\approx2^{168.33}$\\
\hline
\end{tabular}
\end{center}
Here $T=q+2L$, so a threshold $T\ge T_0$ implies
$q\ge T_0-2L$ for the adversary's own swap queries. For standard
XOR queries it implies $q\ge(T_0-2L)/2$. These estimates make the
output-length dependence explicit and recover the expected quantum
preimage and collision exponents whenever that length overhead is
negligible at the relevant scale.

These are concrete bounds for the SHA-3 sponge parameters in the random
permutation model. They do not establish that the fixed Keccak permutation
itself behaves as a uniformly random permutation. However, they indicate that the cryptanalysis of SHA3 is unlikely to be affected by even large scale quantum computers, unless they exploit significant structure in the Keccak permutation.

\subsection{Davies--Meyer}
\label{subsec:explicit-davies-meyer}

Let $N=2^n$ and let $\Phi=(\varphi_k)_{k\in\mathcal K}$ be an ideal
cipher on $\{0,1\}^n$. Recall that
\[
 \mathsf{DM}^{\Phi}(k,x)=\varphi_k(x)\oplus x.
\]
A fixed-target preimage of $w$ is a pair $(k,x)$ with
$\varphi_k(x)=x\oplus w$. A collision consists of distinct inputs
$(k,x)\ne(k',x')$ with the same Davies--Meyer output. Keys may be
chosen by the adversary and queried in superposition.

\begin{theorem}[Explicit Davies--Meyer security]
\label{thm:explicit-davies-meyer}
For a $q$-query algorithm, let $T_{\mathrm{pre}}=q+1$ and
$T_{\mathrm{col}}=q+3$. Whenever the relevant $T$ satisfies
$3T-1\le N$, its success probability is bounded by
\begin{align}
\label{eq:dm-pre-exact}
 p_{\mathrm{pre}}
 &\le\left[
 \frac{2\sqrt2\,a_N(T_{\mathrm{pre}})T_{\mathrm{pre}}}{\sqrt N}
 +e_N(T_{\mathrm{pre}})\right]^2,\\
\label{eq:dm-col-exact}
 p_{\mathrm{col}}
 &\le\left[
 \frac{4\sqrt2\,a_N(T_{\mathrm{col}})T_{\mathrm{col}}^{3/2}}
      {3\sqrt N}
 +e_N(T_{\mathrm{col}})\right]^2.
\end{align}
For $N\ge16$ and $1\le T\le\sqrt N/4$, these imply respectively
\begin{align}
\label{eq:dm-pre-simple}
 p_{\mathrm{pre}}&\le\frac{16384(q+1)^2}{2^n},\\
\label{eq:dm-col-simple}
 p_{\mathrm{col}}&\le
 \frac{16(q+3)^3+16384(q+3)^2}{2^n}.
\end{align}
There is no dependence on the number of keys. All statements also hold
for a single random permutation, by taking a singleton key space.
\end{theorem}

\begin{proof}
For preimages, use the increasing property that some recorded triple
$(k,x,y)$ satisfies $x\oplus y=w$. For a fixed key and a forward
input, exactly one output can satisfy this equation; for an inverse
input, exactly one preimage can satisfy it. Hence $s_t\le1$.
An algorithm outputting $(k,x)$ can supply the claimed triple
$(k,x,x\oplus w)$ without another query. Apply
\Cref{thm:explicit-search} with $\ell=1$ and bound its sum by
$q/\sqrt{N-q+1}$. Enlarging $q$ to $T_{\mathrm{pre}}$ gives
\eqref{eq:dm-pre-exact}.

For collisions, the increasing property is that two distinct recorded
triples have equal values of $x\oplus y$. In a nonsatisfying database
of total size at most $t$, there are at most $t$ such values. For any
fixed key and forward input, each value specifies at most one output
which creates a collision. The same holds for inverse queries.
Consequently $s_t\le t$, even when the colliding triples use different
keys.

An ordinary collision output $(k,x),(k',x')$ need not include the common
hash value. Append one forward query to obtain $y=\varphi_k(x)$, and
supply the two claimed triples
\[
 (k,x,y),\qquad(k',x',y\oplus x\oplus x').
\]
For distinct inputs, these form a consistent collision witness exactly
when the original output is a collision. Apply
\Cref{thm:explicit-search} with $q'=q+1$ and $\ell=2$. Using
$\sum_{t=0}^{q'-1}\sqrt t\le\frac23(q')^{3/2}$ and enlarging the
budget to $T_{\mathrm{col}}$ proves \eqref{eq:dm-col-exact}.

Finally, \eqref{eq:application-error-simplification} bounds the preimage
amplitude by $92T/\sqrt N$, whose square is less than
$16384T^2/N$. For collisions it bounds the amplitude by
$((8/3)T^{3/2}+88T)/\sqrt N$. Applying
$(u+v)^2\le2u^2+2v^2$ gives
\[
 p_{\mathrm{col}}
 \le\frac{(128/9)T^3+15488T^2}{N}
 \le\frac{16T^3+16384T^2}{N}.
\]
\end{proof}

Thus constant success requires $\Omega(2^{n/2})$ queries for a fixed-target
Davies--Meyer preimage and $\Omega(2^{n/3})$ queries for a collision.
The improved soundness error contributes only a quadratic term in the
collision bound.

\section{Statement on AI Use}
The core idea of viewing the compressed oracle as a POVM on the standard purification is due to the authors, as well as the main applications. GPT-6 Astra was used throughout the creation of this manuscript, in particular the calculations establishing the intertwining bound were performed by Astra, as well as working out the explicit constants for SHA-3 parameters in \Cref{subsec:explicit-sha3}. The technical overview and introduction were prepared by an interlaced combination of text from Astra and the authors. Astra further helped format and prepare the bibliography. For instance, a representative prompt is, after defining all relevant notation and machinery in context, ``Compute the spectrum of $G_r$.'' The authors then verified and assembled the output, as well as providing intuition and perspective.

\bibliographystyle{alpha}
\bibliography{ref}

\end{document}